\documentclass[final]{IEEEtran}

\usepackage[pdftex]{graphicx}
\usepackage{epstopdf}
\usepackage[cmex10]{amsmath}
\usepackage{amssymb}
\usepackage{enumerate}
\usepackage{algorithm}
\usepackage{algpseudocode}
\usepackage{empheq}
\usepackage{eqparbox}
\usepackage{amsfonts}
\usepackage{amsthm}
\usepackage{url}
\usepackage{color}
\usepackage{cite}
\usepackage{acro}

\newcommand{\be}{\begin{equation}}
\newcommand{\ee}{\end{equation}}
\newcommand{\ba}{\begin{array}}
\newcommand{\ea}{\end{array}}
\newcommand{\bea}{\begin{eqnarray}}
\newcommand{\eea}{\end{eqnarray}}
\newcommand{\herm}{^{\mbox{\scriptsize H}}}

\newcommand{\vbar}{\raisebox{.17ex}{\rule{.04em}{1.35ex}}}
\newcommand{\vbarind}{\raisebox{.01ex}{\rule{.04em}{1.1ex}}}
\newcommand{\R}{\ifmmode {\rm I}\hspace{-.2em}{\rm R} \else ${\rm I}\hspace{-.2em}{\rm R}$ \fi}
\newcommand{\T}{\ifmmode {\rm I}\hspace{-.2em}{\rm T} \else ${\rm I}\hspace{-.2em}{\rm T}$ \fi}
\newcommand{\N}{\ifmmode {\rm I}\hspace{-.2em}{\rm N} \else \mbox{${\rm I}\hspace{-.2em}{\rm N}$} \fi}
\newcommand{\B}{\ifmmode {\rm I}\hspace{-.2em}{\rm B} \else \mbox{${\rm I}\hspace{-.2em}{\rm B}$} \fi}
\newcommand{\Hil}{\ifmmode {\rm I}\hspace{-.2em}{\rm H} \else \mbox{${\rm I}\hspace{-.2em}{\rm H}$} \fi}
\newcommand{\C}{\ifmmode \hspace{.2em}\vbar\hspace{-.31em}{\rm C} \else \mbox{$\hspace{.2em}\vbar\hspace{-.31em}{\rm C}$} \fi}
\newcommand{\Cind}{\ifmmode \hspace{.2em}\vbarind\hspace{-.25em}{\rm C} \else \mbox{$\hspace{.2em}\vbarind\hspace{-.25em}{\rm C}$} \fi}
\newcommand{\Q}{\ifmmode \hspace{.2em}\vbar\hspace{-.31em}{\rm Q} \else \mbox{$\hspace{.2em}\vbar\hspace{-.31em}{\rm Q}$} \fi}
\newcommand{\Z}{\ifmmode {\rm Z}\hspace{-.28em}{\rm Z} \else ${\rm Z}\hspace{-.28em}{\rm Z}$ \fi}

\newcommand{\ds}{\displaystyle}
\newtheorem{thm}{Theorem}

\newtheorem{cor}{Corollary}

\newtheorem{remark}{Remark}

\newcommand{\M}[1]{\mathbf{#1}}

\newcommand{\rev}[1]{{\color{black} #1}}

\DeclareMathOperator*{\maxmin}{max\,min}

\begin{document}
	\vspace{-0.8cm}
\title{Multi-antenna Interference Management for Coded Caching}%
\author{
Antti T\"olli, \IEEEmembership{Senior Member, IEEE,}
Seyed Pooya Shariatpanahi, 
Jarkko Kaleva, \IEEEmembership{Member, IEEE} and 
Babak Khalaj, \IEEEmembership{Member, IEEE}
\thanks{	This work was supported in part by the Academy of Finland grants No. 279101 and 319059, as well as 6Genesis Flagship grant No. 318927. B. Hossein Khalaj was supported in part by a grant from the Institute for
Research in Fundamental Sciences (IPM) and in part by the Iran National
Science Foundation under Grant 97022420.
	Parts of this work has been published in 2018 IEEE International Symposium of Information Theory, 2018 International Workshop on Contenct Caching and Delivery in Wireless Networks and 2018 Asilomar Conference on Signals, Systems and Computers. 
	
	A. T\"olli and J. Kaleva are with Centre for
	Wireless Communications, University of Oulu P.O. Box 4500, FIN-90014 University of Oulu, Finland \{antti.tolli, jarkko.kaleva\}@oulu.fi. 
	
	S. P. Shariatpanahi is with the School of
Electrical and Computer Engineering, College of Engineering, University of
Tehran, Tehran IR 143995-7131, Iran (p.shariatpanahi@ut.ac.ir). 

B. Khalaj is with the Department of Electrical Engineering, Sharif University of Technology,
Tehran 11365-11155, Iran, and School of Computer Science, Institute for
Research in Fundamental Sciences, Tehran 19395-5746, Iran (e-mail:
khalaj@sharif.ir).
 
 }
}

\maketitle
	\vspace{-1.8cm}
\begin{abstract}
A multi-antenna broadcast channel scenario is considered where a base station delivers contents to cache-enabled user terminals. A joint design of coded caching (CC) and multigroup multicast beamforming is proposed to benefit from spatial multiplexing gain, improved interference management and the global CC gain, simultaneously. The developed general content delivery strategies utilize the multiantenna multicasting opportunities provided by the CC technique while optimally balancing the detrimental impact of both noise and inter-stream interference from coded messages transmitted in parallel.
Flexible resource allocation schemes for CC are introduced where the multicast beamformer design and the receiver complexity are controlled by varying the size of the subset of users served during a given time interval, and the overlap among the multicast messages transmitted in parallel, indicated by parameters $\alpha$ and  $\beta$, respectively. 
Degrees of freedom (DoF) analysis is provided showing that the DoF only depends on $\alpha$ while it is independent of $\beta$.
The proposed schemes are shown to provide the same degrees-of-freedom at high signal-to-noise ratio (SNR) as the state-of-art methods and, in general, to perform significantly better, especially in the finite SNR regime, than several baseline schemes. %
\end{abstract}

\section{Introduction}

Video delivery will be responsible for about $80$ percent of the mobile traffic by 2021 according to the Cisco traffic forecast report \cite{Cisco2019}, which draws attention to the content caching technology as a key element of next generation networks. 
Content caching involves prefetching most popular contents at network edge during low-congested hours mitigating network overcrowding when the real requests of users will show up. This idea has been widely investigated in various wireless network scenarios such as using cache-enabled helpers \cite{Shanmugam2013}, device-to-device collaboration \cite{GolrezaeiD2D2012, JiD2D2015}, small cell networks \cite{Bastug2014}, multi-hop networks \cite{Gitzenis2013}, and Cooperative Multi-Point (CoMP) \cite{Liu2014}.

While the above works clearly demonstrate the benefits of caching in wireless networks%
, the pioneering work of \cite{MaddahAli-2014} considers an information theoretic framework for the caching problem, through which a novel \emph{coded caching} (CC) scheme is proposed. In the coded caching scheme the idea is that, instead of simply replicating high-popularity contents near-or-at end-users (at the cache content placement phase), one should spread different contents at different caches. This way, at the content delivery phase, common coded messages could be broadcast to different users with different demands, that would benefit all of the users resulting in substantial gains in large networks. This \emph{global caching gain} relies on the observation that almost in all communication scenarios, broadcasting is much simpler than unicasting. %
Also, as was proven later in~\cite{Wan-Tuninetti-Piantanida-ITW16,Yu2018}, the performance of this CC scheme is optimal under the assumption of uncoded prefetching, i.e. when coding is allowed only at the delivery phase. Follow-up works extend the coded caching scheme proposed in \cite{MaddahAli-2014} to other setups such as online coded caching \cite{Pedarsani2016}, hierarchical coded caching \cite{Karamchandi2016},  and multi-server scenarios \cite{Shariatpanahi2016}.  All these works suggest that the same kind of CC gain is achievable under various network models.

In order to examine the CC approach in wireless networks the specific characteristics of wireless medium (such as the broadcast nature, fading, and interference) must be investigated to be able to implement the original idea of~\cite{MaddahAli-2014} in mobile delivery scenarios. In order to achieve this goal, in this paper, we investigate the potentials of applying CC to a single-cell multiple-input single-output (MISO) broadcast channel (BC). In such a scenario a multi-antenna base station (BS) transmitter, which has access to the contents library, satisfies content requests of single-antenna users (mobile devices) via a shared wireless medium. The users are cache-enabled, and thus, before the delivery phase begins, they have cached relevant data from the library during off-peak hours. We focus on a joint design of the beamforming scheme used at the BS and the CC design of multicast messages such that the achievable delivery rate is maximized in finite SNR regime. The main goal of our paper is to employ the multiple antennas at the transmitter to manage the interaction between noise and interference between coded messages (i.e., inter-stream interference) and at the same time to benefit from the gains promised by the CC paradigm. 

\subsection{Related Work}

In the context of benefiting from CC gains in wireless networks, the authors in~\cite{Zhang2017} consider the effect of delayed channel state information at the transmitter (CSIT) and demonstrate a synergy between CSIT and caching. Moreover, the work~\cite{Naderalizadeh2017} investigates wireless interference channels where both the transmitters and receivers are cache-enabled. They show that, considering one-shot transmission schemes, the caches at receive and transmit sides are of equal value in the sense of network DoF, which is also confirmed to be the case in cellular networks~\cite{Naderalizadeh2017-2}. In contrast,~\cite{Tahmasbi2017} treats the same setup with mixed-CSIT and unveils the importance of receiver side memory in such a scenario. Cache-enabled interference channels are also investigated by other works such as~\cite{MaddahIC2015, Cao2017,Hachem2016,Roig2017} which do not restrict the schemes to be one-shot, and thus benefit from practically more complex interference alignment (IA) schemes. Also, the authors in~\cite{Ji2016} investigate CC schemes in wireless device-to-device networks and adapt the original CC scheme to a server-less setup, while~\cite{Shabani2016} shows the benefit of device mobility in such scenarios. Furthermore, the cache-enabled cloud radio access networks (C-RAN) are studied in~\cite{Zhang2018}. 

All the aforementioned papers consider wireless networks in the high signal-to-noise-ratio (SNR) regime, expressing their performance in terms of degrees-of-freedom (DoF). As high SNR analysis is not always a good indicator for practical implementations performance, there is still a gap which should be filled in with finite SNR analysis of the CC idea. The papers~\cite{Ngo2017} and~\cite{Shariatpanahi2017} propose different CC schemes in a wireless MISO-BC model, and provide a finite SNR analysis, in different system operating regimes. While the main idea in~\cite{Ngo2017} is to use rate-splitting along with CC, the authors in \cite{Shariatpanahi2017} propose a joint design of CC and zero-forcing (ZF) to benefit from the spatial multiplexing gain and the global gain of CC, at the same time. While the ideas in \cite{Shariatpanahi2017} originally came from adapting the multi-server CC scheme of \cite{Shariatpanahi2016} (which is almost optimal in terms of DoF as shown in \cite{Naderalizadeh2017}) to a Gaussian MISO-BC, the interesting observations in \cite{Shariatpanahi2017} reveal that careful code and beamformer design modifications have significant effects on the finite SNR performance. 

Moreover, it should be noted that \cite{Piovano2017} also considers using the rate-splitting along with CC and propose schemes benefiting from spatial multiplexing and CC gains in a MISO-BC setup. However, as shown in \cite{Shariatpanahi-Caire-Khalaj-TIT18} the resulting DoF performance is worse than the zero-forcing proposal in \cite{Shariatpanahi2017}, and, consequently, is  inferior to our scheme as well. Although the works \cite{Amiri2017} and \cite{Bidokhti2017} consider the finite SNR performance of coded caching in broadcast channels, they assume a single-antenna transmitter, and thus in contrast to our paper, the interference management potentials of transmitter via its multiple antennas are not investigated. 
Finally, the authors in~\cite{Lampiris-Elia-2018} addressed the subpacketization bottleneck in the multicast CC schemes \cite{MaddahAli-2014,Shariatpanahi2016, Shariatpanahi2017}, and proposed a simple ZF based multiantenna transmission scheme substantially reducing the required subpacketization, while providing the same high SNR DoF as~\cite{Shariatpanahi2016, Shariatpanahi2017}.

\subsection{Main Contributions}

In this paper, extending the joint interference nulling and CC concept originally proposed in~\cite{Shariatpanahi2017,Shariatpanahi-Caire-Khalaj-TIT18}, a joint design of CC and generic multicast beamforming is introduced to simultaneously benefit from spatial multiplexing gain, improved management of inter-stream interference from coded messages transmitted in parallel, and the global caching gain. Our proposal results in a general content delivery scheme for any values of the problem parameters, i.e., the number of users $K$, library size $N$, cache size $M$, and number of transmit antennas $L$ such that $t=KM/N$ is an integer value.
Due to the ZF beamforming constraint, the number of antennas is restricted to be strictly $L\leq K-t$ in~\cite{Shariatpanahi2017,Shariatpanahi-Caire-Khalaj-TIT18} as the null space beamformer is unique only when $L\leq K-t$. However, the generic multicast beamformer design introduced in this paper can be designed to fully benefit from the the interference free signal space even when $L>K-t$. 
The general signal-to-interference-plus-noise ratio (SINR) expressions are handled directly to optimally balance the detrimental impact of both noise and inter-stream interference at low SNR. As the resulting optimization problems are not necessarily convex, successive convex approximation (SCA) of non-convex SINR constraints, similarly to~\cite{Venkatraman-Tolli-Juntti-Tran-TSP17}, are used to devise efficient iterative algorithms. \rev{In this paper, the 
multigroup multicast beamformer design is also generalized to any combination of \textit{overlapping user groups} while assuming successive interference cancellation (SIC) receiver}.

\rev{The generalized flexible multigroup multicast beamformer design proposed in this paper allows us to introduce innovative flexible resource allocation schemes for coded caching. Depending on the spatial degrees of freedom and the available SNR, a varying number of multicast  messages can be transmitted in parallel to distinct subsets of users.
Instead of always serving a group of $t+L$ users as in~\cite{Shariatpanahi2017,Shariatpanahi-Caire-Khalaj-TIT18}, \textit{the size of the user subset served during a given time interval} is controlled by a parameter $\alpha$ such that $t+\alpha \leq t+L$.
Furthermore, the parameter $\beta$ is introduced to control the \textit{overlap among the multicast messages} transmitted in parallel. It defines how many parallel messages the receiver should be able to distinguish from each other using the SIC receiver.
The benefits of new $\alpha$- and $\beta$-schemes are} twofold. First, the complexity of the beamformer design \rev{at the transmitter} is managed by controlling the number of constraints and variables in the corresponding optimization problem. \rev{The receiver complexity depends on the number of parallel messages to be decoded by each user, which is controlled by the parameter $\beta$. In the least complex form of implementation with $\beta=1$, the multicast messages do not overlap at all.  This results in linear receiver implementation, which does not require SIC unlike in the general case. Second, by using $\alpha<L$, a better rate performance can be attained at low to medium SNR range} by exploiting the transmit antennas to achieve multiplexing gain and at the same time compensating for the worst users channel effects. 
 Thus, with a \rev{modest loss in performance only at high SNR region}, the complexity of both the receiver and transmitter implementation can be significantly reduced\rev{, compared to the generalized multicast beamformer design with $\alpha=\beta=L$}. Finally, DoF analysis of the proposed schemes is provided showing that the DoF only depends on $\alpha$ and it is independent of $\beta$.

Parts of this paper have been published in the conference
publications~\cite{Tolli-Shariatpanahi-Kaleva-Khalaj-ISIT18,Tolli-Shariatpanahi-Kaleva-Khalaj-WiOPT18,Tolli-Shariatpanahi-Kaleva-Khalaj-Asilomar18}. Considering simple $3$- and $4$-user scenarios, the basic idea of combining multi-group multicast beamformer design and CC was first introduced in~\cite{Tolli-Shariatpanahi-Kaleva-Khalaj-ISIT18}, while the reduced complexity multicast mode selection idea and the simple linear multicast beamforming strategy were introduced in~\cite{Tolli-Shariatpanahi-Kaleva-Khalaj-WiOPT18} and~\cite{Tolli-Shariatpanahi-Kaleva-Khalaj-Asilomar18}, respectively. 
In this paper, in addition to simple scenarios considered in~\cite{Tolli-Shariatpanahi-Kaleva-Khalaj-ISIT18,Tolli-Shariatpanahi-Kaleva-Khalaj-WiOPT18,Tolli-Shariatpanahi-Kaleva-Khalaj-Asilomar18}, a general content delivery scheme applicable for a wide range of the problem parameters values is provided along with a corresponding DoF analysis. \rev{Furthermore, all Matlab codes to regenerate the results of this paper are available online at \url{https://github.com/kalesan/sim-cc-miso-bc}.}

In this paper we use the following notations. We use $(.)\herm$ to denote the Hermitian of a complex matrix. Let $\C$ and $\N$ denote the set of complex and natural numbers and $\|.\|$ be the norm of a complex vector. Also $[m]$ denotes the set of integer numbers
$\{1, . . . , m\}$, and $\oplus$ represents addition in the corresponding finite field. For any vector $\M{v}$, we define $\M{v}^{\perp}$ such that $\M{v}\herm \M{v}^{\perp} = 0$.  Moreover, $\mathcal{A}$ and $|\mathcal{A}|$ denote a set of indexes and its cardinality, while a collection of sets and the number of such sets are indicated by $\mathsf{B}$ ans $|\mathsf{B}|$, respectively.

\section{System Model}
\label{sec:ibc_sysmodel}

Downlink transmission from a single $L$-antenna BS serving $K$ cache enabled single-antenna users is considered. The BS is assumed to 
have access to a library of $N$ files
$\{W_1, \ldots, W_N\}$, each of size $F$ bits.
Each user $k$ is equipped with a cache memory of $MF$ bits and has a message $Z_k = Z_k(W_1, \ldots, W_N)$ stored in its cache, where $Z_k(\cdot)$ denotes a function of the library files with entropy not larger than $MF$ bits.  This operation is referred to as the {\em cache content placement}, and it is performed once and at no cost, e.g. during network off-peak hours.

Upon a set of requests \rev{$d_k \in [N]$} at the \emph{content delivery} phase, the BS multicasts coded signals, such that at the end of transmission all users can reliably decode
their requested files. Notice that user $k$ decoder, in order to produce the decoded file $\widehat{W}_{d_k}$, makes use of its own cache content $Z_k$ as well as its own received signal from the wireless channel.

The received signal at user terminal $k$ at time instant $i, i = 1,\ldots, n$ can be written as
\begin{equation}
\label{eq:recvsig}
{y}_k(i) = \M{h}_{k}\herm (i)\sum_{\mathcal{T}\subseteq \mathcal{S}} \M{w}_\mathcal{T}^\mathcal{S}(i) \tilde{X}_\mathcal{T}^\mathcal{S}(i)  + {z}_k(i),
\end{equation}
where the channel vector between the BS and UE $k$ is denoted by $\M{h}_{k} \in \mathbb{C}^{	L}$,  $ \M{w}_\mathcal{T}^\mathcal{S}$ is the multicast beamformer dedicated to users in subset $\mathcal{T}$ of set \rev{$\mathcal{S} \subseteq [K]$} of users, and $\tilde{X}_\mathcal{T}^\mathcal{S}(i)$ is the corresponding multicast message chosen from a unit power complex Gaussian codebook at time instant $i$. The size of $\mathcal{T}$ depends on the parameters $K$, $M$ and $N$ such that $|\mathcal{T}|=t+1$, where $t\triangleq KM/N$~\cite{Shariatpanahi2016,Shariatpanahi-Caire-Khalaj-TIT18}. The main idea in CC is (by careful cache content placement) to provide multicasting opportunities to groups of size $t+1$, in which a common coded message would be useful for all the members of the multicast group. This is called the \emph{Global Coded Caching Gain}, and is proportional to the total memory of the users, i.e., $KM$, normalized by the library size, i.e., $N$ (for more details refer to \cite{MaddahAli-2014}). In the following, the time index $i$ is ignored for simplicity. The receiver noise is assumed to be 
circularly symmetric zero mean ${z}_k \sim \mathcal{CN}(0, N_0)$. Finally, the CSIT of all $K$ users is assumed to be perfectly known at the BS. 
Note that \eqref{eq:recvsig} is defined for a \rev{given set of users $\mathcal{S}$} served at time instant $i$. Depending on the chosen transmission strategy and parametrization, the delivery of the requested files ${W}_{d_k} \ \forall \ k$ may require multiple time intervals/slots carried out for all possible partitionings and subsets \rev{$\mathcal{S}\subseteq [K]$}. 

\section{Multicast Beamforming for Coded Caching}\label{sec:multicast_bf_cc}

In this work, we focus on the worst-case (over the users) delivery rate at which the system can serve all users requesting for any file of the library. Multicasting opportunities due to the coded caching~\cite{MaddahAli-2014,Shariatpanahi2016,Shariatpanahi2017} are utilized to devise an efficient multiantenna multicast beamforming method that perform well over the entire SNR region. 

\rev{Before 
presenting a high-level description of our proposed scheme, 
let us first provide a brief review of the original coded caching scheme proposed in \cite{MaddahAli-2014}, and then describe the multi-server and multi-antenna extensions in \cite{Shariatpanahi2016}, and \cite{Shariatpanahi2017,Shariatpanahi-Caire-Khalaj-TIT18}, respectively. 
Assuming $t=KM/N \in \mathbb{N}$, the scheme proposed in~\cite{MaddahAli-2014} first divides each file $W_n$ into ${K \choose t}$ subfiles (i.e., $W_n=\{W_{n,\tau}, \tau \subset [K], |\tau|=t\}$), then user $k$ caches all subfiles $W_{n,\tau}$ in which $k \in \tau,  \ \forall \ n$. In the delivery phase, it can be easily seen that for each $t+1$ subset of users $\mathcal{T}$, a common coded multicast message $\oplus_{k \in \mathcal{T}} W_{d_k,T \backslash k}$ would benefit all the users in the subset $\mathcal{T}$ with providing them one subfile of their requested file~\cite{MaddahAli-2014}. If all such multicast coded messages are successfully delivered, then it can be shown that the missing subfiles will be received at all the users. This will result into the total normalized delay of ${K \choose t+1} / {K \choose t} = (K-t) / (t+1)$~\cite{MaddahAli-2014}.

When
$L$ servers have access to the library, instead of a single server, they can collaboratively send coded messages each of which would benefit $t+L$ users
~\cite{Shariatpanahi2016}. This means the global caching gain (proportional to $t$) and the collaboration multiplexing gain (proportional to $L$) are additive. This will directly reduce the delay by a multiplicative factor of $(t+1)/(t+L)$, which results in the normalized delay of $(K-t)/(t+L)$. The scheme proposed in \cite{Shariatpanahi2016} combines zero-forcing gain and multicasting opportunities provided by coded caching as follows. Assume a subset of users of size $t+L$. Then according to the proposal in \cite{MaddahAli-2014}, a coded common message can be constructed to serve each $t+1$ subset of these users. Now by combining all ${t+L \choose t+1}$ such coded messages, with each message directed in the null space of $L-1$ undesired users, all of them can be sent simultaneously. However, any user $k$ will have to decode $m_k=\binom{t+L-1}{ t }$ different messages, which grows exponentially when $K$, $L$, $N$ are increased with the same ratio (linearly if $t=KM/N=1$). 

The multiserver scheme~\cite{Shariatpanahi2016} is adapted to a multiple-antenna transmitter delivering files to users via a wireless medium in \cite{Shariatpanahi2017,Shariatpanahi-Caire-Khalaj-TIT18}, where each coded message destined to $t+1$ users is nulled at $L-1$ non-desired users. 
For each user $k$, the received signal $y_k$ contains $m_k$ desired signals. Hence, from the receiver perspective, it appears as $m_k$-dimensional Gaussian multiple access channel (MAC) with a feasible rate region defined by $2^{m_k}-1$ rate constraints. Thus, a SIC structure has to be used at each receiver to decode the intended messages~\cite{Shariatpanahi2017,Shariatpanahi-Caire-Khalaj-TIT18}.


While adapting the multi-server coded caching scheme to the wireless multiple-antenna setup was shown to improve the rate of the system also at finite SNR \cite{Shariatpanahi-Caire-Khalaj-TIT18}, compared to non-coded schemes, using the ZF vectors at finite SNR is a highly sub-optimal strategy in general. At finite SNR, the detrimental impact of both inter-stream interference and noise needs to be balanced to arrive at the best performance. Thus, in our proposed scheme we address the challenge of finding the optimum multigroup multicast beamforming vectors with any combination of overlapping user groups for minimizing the delivery time while allowing controlled inter-stream interference among multicast messages transmitted in parallel. Although this modification will significantly improve the system performance especially at low SNR, it introduces a beamformer optimization problem which has significantly higher computational complexity than the simple ZF scheme used in~\cite{Shariatpanahi2017,Shariatpanahi-Caire-Khalaj-TIT18}. 

The aforementioned generalized multigroup multicast beamformer design allows us to introduce innovative flexible resource allocation schemes for coded caching, depending on the available transmit power or computational complexity constraints. 
Instead of always serving a group of $t+L$ users as in~\cite{Shariatpanahi2017,Shariatpanahi-Caire-Khalaj-TIT18}, \textit{the size of the user subset served during a given time interval} can be controlled by a parameter $\alpha$ such that $t+\alpha \leq t+L$.
Furthermore, the parameter $\beta$ is introduced to control the \textit{overlap among the multicast messages} transmitted in parallel. It defines how many parallel messages the receiver should be able to distinguish from each other using the  SIC receiver. Thus, the complexity of both the receiver and transmitter implementation can be significantly reduced without significant , compared to the generalized multicast beamformer design with fully overlapping multicast messages, $\alpha=\beta=L$.}


In \rev{ the rest of} this section, we first introduce the proposed concept and its variations in four simple scenarios and discuss the generalization of the proposed schemes in Section~\ref{sec:general}. The multigroup multicast beamformer design for the classical 3-user case~\cite{MaddahAli-2014,Shariatpanahi2016,Shariatpanahi2017} is first described in \textit{Scenario 1}, which in turn is extended to 4-user case in \textit{Scenario 2} to demonstrate how the size and complexity of the problem quickly increases for larger values of $K$. Reduced complexity beamformer design alternatives to \textit{Scenario 2} are introduced in \textit{Scenarios 3 and 4} by controlling the size \rev{$\alpha$} of the subset $\{\mathcal{S} \subseteq [K]\}$ served during a given time interval, and the overlap \rev{$\beta$} among the multicast messages transmitted in parallel, respectively.

\subsection{Scenario 1: $L\geq2$, $K=3$, $N=3$ and $M=1$}\label{sec:scenario1}
Consider a content delivery scenario illustrated in Fig.~\ref{fig:L2_K3_scenario1}, where a transmitter with $L\geq2$ antennas should deliver requests arising at $K=3$ users from a library $\mathcal{W}=\{A, B, C\}$ of size $N=3$ files each of $F$ bits. Suppose that in the cache content placement phase each user can cache $M=1$ files of $F$ bits, without knowing the actual requests beforehand. In the content delivery phase we suppose each user requests one file from the library. Following the same cache content placement strategy as in \cite{MaddahAli-2014} the cache contents of users are as follows
\begin{equation}\nonumber
Z_1 = \{A_1,B_1,C_1\}, Z_2 = \{A_2,B_2,C_2\}, Z_3 = \{A_3,B_3,C_3\}
\end{equation}
where each file is divided into 3 equal-sized subfiles.
\begin{figure}
	\centering
	\includegraphics[width=\columnwidth]{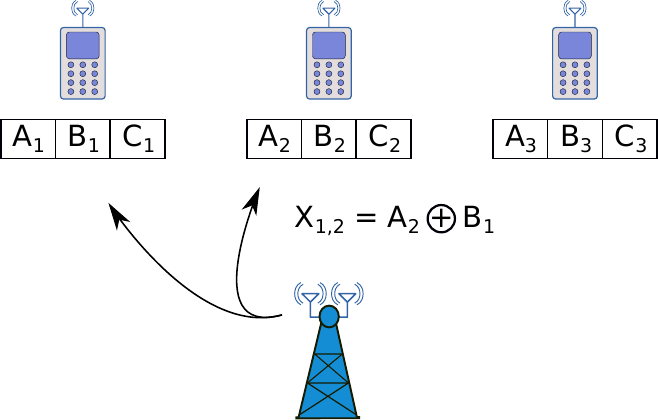}
	\caption{\emph{Scenario 1}:  $L=2$, $K=3$, $N=3$ and $M=1$}
	\label{fig:L2_K3_scenario1}
\end{figure}

At the content delivery phase, suppose that the 1st, the 2nd, and the 3rd user request files $A$, $B$, and $C$, respectively. In the simple broadcast scenario in \cite{MaddahAli-2014}, the following coded messages are sent to users $\mathcal{S}= \{1,2,3\}$  by the transmitter one after another
\begin{equation}
X_{1,2} = A_2 \oplus B_1, \
X_{1,3} = A_3 \oplus C_1, \
X_{2,3} = B_3 \oplus C_2
\end{equation}
where $\oplus$ represents summation in the corresponding finite field, and the superscript $\mathcal{S}$ is omitted for ease of presentation. In such coding scheme, each coded message is received by all $3$ users, but is only beneficial to $2$ of them. For example, $X_{1,2}$ is useful for the 1st and 2nd user only. %
It can be easily checked that after transmission is concluded, all users can decode their requested files. Moreover, for every possible combination of the users requests, the scheme works with the same cache content placement, but with another set of coded delivery messages.

Consequently, in \emph{Scenario 1}, we can combine the spatial multiplexing gain and the global caching gain following the scheme in~\cite{Shariatpanahi2017} (see also~\cite{Shariatpanahi2016,Naderalizadeh2017}). In~\cite{Shariatpanahi2017}, the unwanted messages at each user are forced to zero by sending
\begin{equation}
\mathbf{h}_3^{\perp}\tilde{X}_{1,2}+\mathbf{h}_2^{\perp}\tilde{X}_{1,3}+\mathbf{h}_1^{\perp}\tilde{X}_{2,3}
\end{equation}
where $\tilde{X}$ stands for the modulated $X$, chosen from a unit power complex Gaussian codebook~\cite{Shariatpanahi2017}. 

The key point here is to note that although this scheme is order-optimal in terms of DoF~\cite{Naderalizadeh2017} it is suboptimal at low SNR regime~\cite{Shariatpanahi2017,Shariatpanahi-Caire-Khalaj-TIT18}. Therefore, in this paper, instead of nulling interference at unwanted users, general multicast beamforming vectors $\M{w}_\mathcal{T}^\mathcal{S}$ are defined as
\begin{equation}
\sum_{\mathcal{T}\subseteq [3],\\ |\mathcal{T}|=2}  \!\!\!\! \M{w}_\mathcal{T}^\mathcal{S} \tilde{X}_\mathcal{T}^\mathcal{S} =\mathbf{w}_{1,2}\tilde{X}_{1,2}+\mathbf{w}_{1,3}\tilde{X}_{1,3}+\mathbf{w}_{2,3}\tilde{X}_{2,3}
\end{equation}
where $[K]$  denotes the set of integer numbers $\{1, . . . , K\}$ and the superscript $\mathcal{S}$ is omitted for simplicity. %
As a result, the received signals at users $1-3$ will be
\begin{align} \nonumber y_1 &=\underline{(\mathbf{h}_1\herm\mathbf{w}_{1,2})\tilde{X}_{1,2}}+\underline{(\mathbf{h}_1\herm\mathbf{w}_{1,3})\tilde{X}_{1,3}}+(\mathbf{h}_1\herm\mathbf{w}_{2,3})\tilde{X}_{2,3} +z_1 \\ \nonumber
y_2&=\underline{(\mathbf{h}_2\herm\mathbf{w}_{1,2})\tilde{X}_{1,2}}+(\mathbf{h}_2\herm\mathbf{w}_{1,3})\tilde{X}_{1,3}+\underline{(\mathbf{h}_2\herm\mathbf{w}_{2,3})\tilde{X}_{2,3}} +z_2 \\ \nonumber
y_3&=(\mathbf{h}_3\herm\mathbf{w}_{1,2})\tilde{X}_{1,2}+\underline{(\mathbf{h}_3\herm\mathbf{w}_{1,3})\tilde{X}_{1,3}}+\underline{(\mathbf{h}_3\herm\mathbf{w}_{2,3})\tilde{X}_{2,3}} +z_3
\end{align}
where the desired terms for each user are underlined. 

Let us focus on user $1$ who is interested in decoding both $\tilde{X}_{1,2}$, and $\tilde{X}_{1,3}$ while $\tilde{X}_{2,3}$ appears as Gaussian interference. Thus, from receiver 1 perspective, $y_1$ is a Gaussian multiple access channel (MAC). Suppose now user 1 can decode \emph{both} of its required messages $\tilde{X}_{1,2}$ and $\tilde{X}_{1,3}$ with the equal rate\footnote{Symmetric rate is imposed to minimize the time needed to receive both messages  $\tilde{X}_{1,2}$, and $\tilde{X}_{1,3}$.} 
\begin{equation}		
R^1_{MAC}=\min (\frac{1}{2}R^1_{Sum}, R^1_1,R^1_2)%
\end{equation} 
where \rev{$R^1_{Sum} = R^1_1+R^1_2$ and} the rates $ R^1_1$ and $R^1_2$ correspond to $\tilde{X}_{1,2}$, and $\tilde{X}_{1,3}$, respectively. 
Thus, the total useful rate is $2R^1_{MAC}$. Since the user 1 must receive the missing $2/3F$ bits ($A_2$ and $A_3$), the time needed to decode file $A$ is
$T_1=\frac{2F}{3} \frac{1}{2R^1_{MAC}}$.
As all the users decode their files \emph{in parallel}, the time needed to complete the decoding process is constrained by the worst user as
\begin{equation}\label{eq:sym_time}
T=\frac{2F}{3} \frac{1}{\min\limits_{k=1,2,3}2R^k_{MAC}}.
\end{equation}
Then, the \emph{Symmetric Rate (\textit{Goodput}) per user} will be 
\begin{equation}\label{eq:sym_rate}
R_{\mathrm sym}=\frac{F}{T}=3\min_{k=1,2,3}R^k_{MAC}
\end{equation}
which, when optimized with respect to the beamforming vectors, can be found as
\begin{equation}
\max_{\mathbf{w}_{2,3}, \mathbf{w}_{1,3}, \mathbf{w}_{1,2}}  \min_{k=1,2,3}R^k_{MAC}.
\end{equation}

Finally, the symmetric rate maximization for $K = 3$ is given as
\begin{align}\nonumber
\label{prob:equal_rate}
&\ds \underset{\substack{R^k, \gamma^k_i, \M{w}_\mathcal{T}, \forall k,i}}{\max}
\ds \min_{k=1,2,3} \min\left( \frac{1}{2} R_\text{sum}^k, R_{1}^k, R_{2}^k \right)\\
&\begin{array}{rl}
\mathrm{s.\ t.} \! \! \! \! 
& \ds R^k_1 \leq \log(1 + \gamma^k_1),  R^k_2 \leq \log(1 + \gamma^k_2), \\ 
& \ds R^k_\text{sum} \leq \log(1 + \gamma^k_1 + \gamma^k_2),\, k = 1,2,3, 
\\
& \ds \gamma^1_1 \leq \frac{|\M{h}_1\herm \M{w}_{1,2}|^2}
{|\M{h}_1\herm \M{w}_{2,3}|^2 + N_0},
\ds \gamma^1_2 \leq \frac{|\M{h}_1\herm \M{w}_{1,3}|^2}
{|\M{h}_1\herm \M{w}_{2,3}|^2 + N_0},\\
& \ds \gamma^2_1 \leq \frac{|\M{h}_2\herm \M{w}_{2,3}|^2}
{|\M{h}_2\herm \M{w}_{1,3}|^2 + N_0},
\ds \gamma^2_2 \leq \frac{|\M{h}_2\herm \M{w}_{1,2}|^2}
{|\M{h}_2\herm \M{w}_{1,3}|^2 + N_0}, \\
& \ds \gamma^3_1 \leq \frac{|\M{h}_3\herm \M{w}_{2,3}|^2}
{|\M{h}_3\herm \M{w}_{1,2}|^2 + N_0},
\ds \gamma^3_2 \leq \frac{|\M{h}_3\herm \M{w}_{1,3}|^2}
{|\M{h}_3\herm \M{w}_{1,2}|^2 + N_0},
\\
&  \sum_{\mathcal{T}\in\{\{1,2\}, \{1,3\}, \{2,3\}\}} \|\M{w}_\mathcal{T}\|^2 \leq \text{SNR} 
\end{array}
\end{align}
which can be equally presented in
an epigraph from as
\begin{equation}
\label{prob:equal_rate_epigraph}
\begin{array}{rl}
\ds \underset{r, \gamma^k, \M{w}_{k,i}}{\max.} & 
\ds r
\\
\mathrm{s.\ t.} 
& \ds r \leq \frac{1}{2}\log(1 + \gamma^k_1 + \gamma^k_2), \ k = 1,2,3,
\\
& \ds r \leq \log(1 + \gamma^k_1),\  r \leq \log(1 + \gamma^k_2) 
\\
& \text{The rest of the constraints as in~\eqref{prob:equal_rate}}.
\end{array}
\end{equation}
Problem~\eqref{prob:equal_rate_epigraph} is non-convex due to the SINR
constraints. Similarly to~\cite{Venkatraman-Tolli-Juntti-Tran-TSP17}, successive convex approximation (SCA) approach can be used to devise an iterative algorithm that is able to converge to a local solution. To begin with, the SINR
constraint for $\gamma^1_1$ can be reformulated as
\begin{align}
\gamma^1_1 &\leq \frac{|\M{h}_1\herm \M{w}_{1,2}|^2}
{|\M{h}_1\herm \M{w}_{2,3}|^2 + N_0}\\
\label{eq:sinr_constraint_before_sca}
|\M{h}_1\herm \M{w}_{2,3}|^2 + N_0 &\leq 
\frac{|\M{h}_1\herm \M{w}_{1,2}|^2 + |\M{h}_1\herm \M{w}_{2,3}|^2 + N_0}
{1 + \gamma^1_1}
\end{align}
\rev{where both sides of the inequality constraint are convex functions (quadratic or quadratic-over-linear) with respect to the optimization variables. In order to make it convex, the R.H.S of~\eqref{eq:sinr_constraint_before_sca} is linearly approximated (lower bounded) as
\begin{align}
\label{eq:sca}
\ds \mathcal{L}(\M{w}_{2,3},& \M{w}_{1,2}, \M{h}_{1}, \gamma^1_1) \triangleq 
\ds \frac{|\M{h}_1\herm \bar{\M{w}}_{1,2}|^2 + |\M{h}_1\herm \bar{\M{w}}_{2,3}|^2 + N_0}{1 + \bar{\gamma}^1_1} \nonumber    \\
&
- 
2\mathbb{R}\left(
\frac{\bar{\M{w}}_{1,2}\herm \M{h}_1 \M{h}_1\herm}{1 + \bar{\gamma}^1_1} 
\left(\M{w}_{1,2} - \bar{\M{w}}_{1,2}\right)\nonumber
\right) 
\\
& \ds - 2\mathbb{R}\left(
\frac{\bar{\M{w}}_{2,3}\herm \M{h}_1 \M{h}_1\herm}{1 + \bar{\gamma}^1_1} 
\left(\M{w}_{2,3} - \bar{\M{w}}_{2,3}\right)
\right)  \nonumber \\
&
+ \frac{|\M{h}_1\herm \bar{\M{w}}_{1,2}|^2 + |\M{h}_1\herm \bar{\M{w}}_{2,3}|^2 + N_0}{(1 + \bar{\gamma}^1_1)^2} 
\left(\gamma^1_1- \bar{\gamma}^1_1\right)
\end{align}
where $\mathbb{R}(\cdot)$ is the real part of the complex valued argument,} $\bar{\M{w}}_{k,i}$ and $\bar{\gamma}_1^1$ denote the fixed values (points of approximation) for the corresponding variables from the previous iteration. 
Using~\eqref{eq:sca} and reformulating the objective in the epigraph form, the approximated problem is written as
\begin{equation}
\label{prob:equal_rate_approx}
\begin{array}{rl}
\ds \underset{r, \gamma^k, \M{w}_\mathcal{T}}{\max.} & 
\ds r
\\
\mathrm{s.\ t.} 
& \ds r \leq 1/2\log(1 + \gamma^k_1 + \gamma^k_2), \\
&
\ds r \leq \log(1 + \gamma^k_1),\  r \leq \log(1 + \gamma^k_2) \, k = 1,2,3, 
\\
& \ds \mathcal{L}({\M{w}_{2,3}, \M{w}_{1,2}, \M{h}_{1}, \gamma^1_1}) \geq 
|\M{h}_1\herm \M{w}_{2,3}|^2 + N_0,\\
&
\ds \mathcal{L}({\M{w}_{2,3}, \M{w}_{1,3}, \M{h}_{1}, \gamma^1_2}) \geq 
|\M{h}_1\herm \M{w}_{2,3}|^2 + N_0,
\\
& \ds \mathcal{L}({\M{w}_{1,3}, \M{w}_{2,3}, \M{h}_{2}, \gamma^2_1}) \geq 
|\M{h}_2\herm \M{w}_{1,3}|^2 + N_0,
\\
&
\ds \mathcal{L}({\M{w}_{1,3}, \M{w}_{1,2}, \M{h}_{2}, \gamma^2_2}) \geq 
|\M{h}_2\herm \M{w}_{1,3}|^2 + N_0, 
\\
& \ds \mathcal{L}({\M{w}_{1,2}, \M{w}_{2,3}, \M{h}_{3}, \gamma^3_1}) \geq 
|\M{h}_3\herm \M{w}_{1,2}|^2 + N_0,\\
&
\ds \mathcal{L}({\M{w}_{1,2}, \M{w}_{1,3}, \M{h}_{3}, \gamma^3_2}) \geq 
|\M{h}_3\herm \M{w}_{1,2}|^2 + N_0,
\\
& \sum_{\mathcal{T}\in\{\{1,2\}, \{1,3\}, \{2,3\}\}} \|\M{w}_\mathcal{T}\|^2 \leq \text{SNR} 
\end{array}
\end{equation}
This is a convex problem that can be readily solved using existing convex
solvers. However, the logarithmic functions require further approximations to be
able to apply the convention of convex programming algorithms. Problem~\eqref{prob:equal_rate_approx} can be
equally formulated as computationally efficient second order cone problem (SOCP). %
To this end, we
note that the sum rate constraint can be bounded as
\begin{equation}\nonumber
r \leq \frac{1}{2}\log(1 + \gamma^k_1 + \gamma^k_2) =
\log(\sqrt{1 + \gamma^k_1 + \gamma^k_2}) \leq \sqrt{1 + \gamma^k_1 + \gamma^k_2}
\end{equation}
Now, the equivalent SOCP reformulation follows as
\begin{equation}
\label{prob:equal_rate_socp}
\begin{array}{rl}
\ds \underset{\tilde{r}, \gamma^k, \M{w}_k}{\max.} & 
\ds \tilde{r}
\\
\mathrm{s.\ t.} 
& \ds {\tilde{r}}^2 \leq 1 + \gamma^k_1 + \gamma^k_2, 
\ds \tilde{r} \leq 1 + \gamma^k_1,\  \tilde{r} \leq 1 + \gamma^k_2 \, k = 1,2,3, 
\\
& \text{The rest of the constraints as in~\eqref{prob:equal_rate_approx} } 
\text{.}
\end{array}
\end{equation}
Finally, a solution for the original problem~\eqref{prob:equal_rate} can be found by solving~\eqref{prob:equal_rate_approx} in an iterative manner using SCA, i.e, by updating the points of approximations $\bar{\M{w}}_{k,i}$ and $\bar{\gamma}_j^l$ in~\eqref{eq:sca} after each iteration. As each difference-of-convex constraint in~\eqref{eq:sinr_constraint_before_sca} is lower bounded by~\eqref{eq:sca}, the monotonic convergence of the objective of~\eqref{prob:equal_rate_approx} is guaranteed. Note that the final symmetric rates are achieved by time sharing between the rate allocations corresponding to different points (decoding orders) in the sum rate region of the MAC channel.

As a lower complexity alternative, a zero forcing solution, denoted as \emph{CC with ZF}, is also proposed\footnote{Note that the null space beamformer is unique only when $L=2$. Generic multicast beamformers can be designed within the interference free signal space when $L>2$ (See Section~\ref{sec:numexp}).}. By assigning $\M{w}_{1,2}=\mathbf{h}_3^{\perp}/\|\mathbf{h}_3^{\perp}\| \sqrt{p_{1,2}}$, $\M{w}_{1,3}=\mathbf{h}_2^{\perp}/\|\mathbf{h}_2^{\perp}\| \sqrt{p_{1,3}}$,
$\M{w}_{2,3}=\mathbf{h}_1^{\perp}/\|\mathbf{h}_1^{\perp}\| \sqrt{p_{2,3}}$, the interference terms are canceled and \eqref{prob:equal_rate} becomes:
\begin{align}
\label{prob:equal_rate_ZF}
&\ds \underset{\substack{R^k, \gamma^k, p_\mathcal{T}}}{\max}
\ds \min_{k=1,2,3} \min\left( \frac{1}{2} R_\text{sum}^k, R_{1}^k, R_{2}^k \right)\\
&\begin{array}{rl}
\mathrm{s.\ t.} 
& \ds R^k_\text{sum} \leq \log(1 + \gamma^k_1 + \gamma^k_2), \\
&
\ds R^k_1 \leq \log(1 + \gamma^k_1),  \ds R^k_2 \leq \log(1 + \gamma^k_2) \ \forall \ k, 
\\
& \ds \gamma^1_1 \leq u_{1,3} p_{1,2},
\ds \gamma^1_2 \leq u_{1,2} p_{1,3},
\ds \gamma^2_1 \leq u_{2,1} p_{2,3},\\
&
\ds \gamma^2_2 \leq u_{2,3} p_{1,2}, 
\ds \gamma^3_1 \leq u_{3,1} p_{2,3},
\ds \gamma^3_2 \leq u_{3,2} p_{1,3},
\\
&  \sum_{\mathcal{T}\in\{\{1,2\}, \{1,3\}, \{2,3\}\}} p_\mathcal{T} \leq \text{SNR} 
\end{array}\nonumber
\end{align}
where $u_{k,i}=|\M{h}_k\herm \mathbf{h}_i^{\perp}|^2/\|\mathbf{h}_i^{\perp}\|^2 N_0$.
This is readily a convex power optimization problem with three real valued variables, and hence it can be solved in an optimal manner.

In the following, three baseline reference cases for the proposed multiantenna caching scheme are introduced.

\subsubsection{1st Baseline Scheme: CC with ZF (equal power)~\cite{Shariatpanahi2017}}
If the multicast transmit powers are made equal, $p_{1,2}=p_{1,3}=p_{2,3}=SNR/3$, the resulting scheme is the same as originally published in~\cite{Shariatpanahi2017}.
\subsubsection{2nd Baseline Scheme: MaxMinSNR Multicasting} \rev{In this case, a single multicast stream is transmitted at a time \emph{without any inter-stream interference}. Thus, three time slots are required to deliver messages $X_{1,2}$, $X_{1,3}$ and $X_{2,3}$. In timeslot 1,} 
the message $X_{1,2}$ is delivered to the users 1 and 2 by sending the signal
\rev{$\mathbf{w}_{1,2} \tilde{X}_{1,2}$}. A single transmit beamformer \rev{$\mathbf{w}_{1,2}$} is found to minimize the time needed for multicasting the common message:\footnote{This multicast maxmin problem is NP-hard in general, but near-optimal solutions can be obtained by a semidefinite relaxation (SDR) approach, see~\cite{Shariatpanahi2017} and the references therein.}
\begin{equation}\label{eq:CCmaxminSNR1}
T_{1,2}=\frac{F / 3}{\maxmin\limits_{\|\mathbf{w}\rev{_{1,2}}\|^2 \leq SNR}\! \left(\log(1+\frac{|\mathbf{h}_1\herm\mathbf{w}_{1,2}|^2}{N_0}),\log(1+\frac{|\mathbf{h}_2\herm\mathbf{w}_{1,2}|^2}{N_0}) \right)}.
\end{equation}
Similarly, the messages $X_{1,3}$ and $X_{2,3}$ should be delivered to the users with corresponding \rev{multicast beamformers $\mathbf{w}_{1,3}$ and $\mathbf{w}_{2,3}$, and delivery} times $T_{1,3}$ and $T_{2,3}$. 
Finally the resulting symmetric rate (Goodput) per user will be
\begin{equation}\label{eq:rate_maxmin}
R_{\mathrm{maxmin}}={F}/({T_{1,2}+T_{1,3}+T_{2,3}}).
\end{equation}
Note that, in this scheme, only the coded caching gain is exploited, while the multiple transmit antennas are used just for the multicast beamforming gain.

\subsubsection{3rd Baseline Scheme: MaxMinRate Unicast}
In this scheme, only the local caching gain is exploited and the CC gain is ignored altogether. The BS simply sends $\min(K,L)$ parallel independent streams to the users at each time instant. All the users can be served in parallel if $L\geq K$. On the other hand, if $L<K$, the users need to be divided into subsets of size $L$ served at distinct time slots. 

Now, let us consider the case $L=2$ and $K=3$, and focus on users 1 and 2 in time slot 1. The transmitted signal to deliver $A_2$ and $B_1$ to users 1 and 2, respectively, is given as
$\mathbf{w}_{1} \tilde{A}_2+\mathbf{w}_{2} \tilde{B}_1$.
Thus, the delivery time of $F/3$ bits is
\begin{equation}\label{eq:unicast_time}
T_{1,2}=\ds\frac{F/3}{\max\limits_{\sum_{k=1,2}\|\mathbf{w}_k\|^2 \leq SNR} \min(R_1,R_2)}
\end{equation}
where
\begin{equation}\label{eq:rate_unicast}
R_k= \log\left(1+\frac{|\mathbf{h}_k\herm\mathbf{w}_k|^2}{\sum_{i\neq k}|\mathbf{h}_k\herm\mathbf{w}_i|^2+N_0}\right) \text{.}
\end{equation}
The minimum delivery time in~\eqref{eq:rate_maxmin} can be equivalently formulated as a maxmin SINR problem and solved optimally.
By repeating the same procedure for the subsets $\{1,3\}$ and $\{2,3\}$, the symmetric rate expression is equivalent to~\eqref{eq:rate_maxmin}.

\subsection{Scenario 2: $L\geq3$, $K=4$, $N=4$ and $M=1$}
In this scenario, the number of users $K$ and files $N$ is further increased in order to demonstrate how the size and complexity of the problem quickly increases for larger values of $K$. We assume that the BS transmitter has $L\geq3$ antennas, and there are $K=4$ users each with cache size $M=1$, requesting files from a library $\mathcal{W}=\{A, B, C, D\}$ of $N=4$ files. Following the same cache content placement strategy as in \cite{MaddahAli-2014} the cache contents of users are as follows
\begin{align} \nonumber
Z_1 &= \{A_1,B_1,C_1,D_1\},
Z_2 = \{A_2,B_2,C_2,D_2\}, \\
Z_3 &= \{A_3,B_3,C_3,D_3\},
Z_4 = \{A_4,B_4,C_4,D_4\}
\end{align}
where here each file is divided into four non-overlapping equal-sized subfiles.

At the content delivery phase, suppose that the users $1-4$ request files $A-D$, respectively. Here, we have $t\triangleq KM/N=1$ and the subsets $\mathcal{S}$ and $\mathcal{T}$ will be of size $4$ %
 and $t+1=2$, respectively (for details see \cite{Shariatpanahi2016,Shariatpanahi-Caire-Khalaj-TIT18} and Section~\ref{sec:general}). %
 Following the approach of \textit{Scenario 1}, the transmit signal vector is
\begin{align}\label{eq:S2_signal_vector}
\sum_{\mathcal{T}\subseteq \mathcal{S}, |\mathcal{T}|=2} \mathbf{w}_\mathcal{T} \tilde{X}_\mathcal{T} &= \mathbf{w}_{1,2} \tilde{X}_{1,2} + \mathbf{w}_{1,3} \tilde{X}_{1,3} + \mathbf{w}_{1,4} \tilde{X}_{1,4} \nonumber
\\
&
+ \mathbf{w}_{2,3} \tilde{X}_{2,3} + \mathbf{w}_{2,4} \tilde{X}_{2,4} + \mathbf{w}_{3,4} \tilde{X}_{3,4}
\end{align}
where
$X_{1,2}= A_2 \oplus B_1, %
X_{1,3} = A_3 \oplus C_1, %
X_{1,4} = A_4 \oplus D_1, %
X_{2,3} = B_3 \oplus C_2, %
X_{2,4} = B_4 \oplus D_2, %
X_{3,4} =  C_4 \oplus D_3.$
%

The received signal at user $k=1,2,3,4$ is written as
\begin{align} \nonumber
y_k &=  \underline{(\mathbf{h}_k\herm\mathbf{w}_{1,2}) \tilde{X}_{1,2}} + \underline{(\mathbf{h}_k\herm\mathbf{w}_{1,3}) \tilde{X}_{1,3}} + \underline{(\mathbf{h}_k\herm\mathbf{w}_{1,4}) \tilde{X}_{1,4}}\\ \nonumber
&+ (\mathbf{h}_k\herm\mathbf{w}_{2,3}) \tilde{X}_{2,3} + (\mathbf{h}_k\herm\mathbf{w}_{2,4}) \tilde{X}_{2,4} + (\mathbf{h}_k\herm\mathbf{w}_{3,4}) \tilde{X}_{3,4} +z_1 %
\end{align}
where, as an example, the desired terms of user $1$ are underlined. As in \emph{Scenario 1}, each user faces a MAC channel, now with three desired signals, three Gaussian interference terms, and one noise term. Suppose that user $k$ can decode each of its desired signals with the rate $R^k_{MAC}$. Consequently, this user receives useful information with the rate  $3R^k_{MAC}$, and the time required to fetch the entire file is $T_1=\frac{3F}{4}\frac{1}{3R^k_{MAC}}$. Following the same steps as in~\eqref{eq:sym_time}--\eqref{eq:sym_rate}, the symmetric rate per user can be found as
\begin{equation}\label{eq:symrate_4users}
R_{\mathrm sym}=\frac{F}{T}=4\max_{\mathbf{w}_\mathcal{T}, \mathcal{T}\subseteq[4], |\mathcal{T}|=2}  \min_{k=1,2,3,4}R^k_{MAC}
\end{equation}
where 
\begin{equation}\label{eq:MAC_4users}
R^k_{MAC} = \min \left(R_1^k, R_2^k, R_3^k, \frac{1}{2}R_4^k, \frac{1}{2}R_5^k, \frac{1}{2}R_6^k, \frac{1}{3}R_7^k\right)
\end{equation}
and where the rate bounds $ R_1^1$, $R^1_2$ and $R^1_3$ of user $1$, for example, correspond to $\tilde{X}_{1,2}$, $\tilde{X}_{1,3}$ and $\tilde{X}_{1,4}$, respectively. The bounds $R^1_4$, $R^1_5$ and $R^1_6$ limit the sum rate of any combination of two transmitted multicast signals, and finally $R^1_7$ is the sum rate bound for all 3 messages. 
%
\rev{The SCA method is again used to solve~\eqref{eq:symrate_4users}, similarly to \eqref{prob:equal_rate}-\eqref{prob:equal_rate_approx}.} 
%
%
%


\subsection{Scenario 3: $L\geq3$, $K=4$, $N=4$, $M=1$ and \rev{$\alpha=2$}}
In this example, a reduced complexity alternative for \textit{Scenario 2} is considered. Instead of fixing the size of the served user set to $|\mathcal{S}|\rev{=\min(t+L,K)}=4$ as in \textit{Scenario 2}, 
we \rev{set $\alpha=2$ and} restrict the size of the subsets $\mathcal{S} \subset [4] $ benefiting from a common transmitted signal to $|\mathcal{S}|\rev{=t+\alpha}=3$.  %
Thus, the size of the MAC channel for each user is reduced from $3$ to $2$ and each user needs to decode just $2$ multicast streams. This in turn, reduces the complexity of the problem for determining the beamforming vectors for each subset $\mathcal{S} \subset [4] $. %
As will be shown later, besides complexity reduction, controlling the size of each subset allows us to handle the trade-off between the multiplexing and multicast beamforming gains due to multiple transmit antennas, resulting in even better rate performance at certain SNR values. Note that for $|\mathcal{S}|=2$ \rev{($\alpha=1$)}, the beamformer design for each subset $\mathcal{S}$ reduces to the baseline max-min SNR scheme (see~\eqref{eq:CCmaxminSNR1} for $K=3$).

The cache content placement works similarly, except that each subfile is split into $2$ mini-files (indicated by superscripts) in order to allow different contents to be transmitted in each subset $\mathcal{S}$.
As a result, the following content is stored in user cache memories
\begin{align} \nonumber
&Z_1 = \{A_1^1,A_1^2,B_1^1,B_1^2,C_1^1,C_1^2,D_1^1,D_1^2\}, \\ \nonumber
&
Z_2 = \{A_2^1,A_2^2,B_2^1,B_2^2,C_2^1,C_2^2,D_2^1,D_2^2\} \\ \nonumber
&Z_3 = \{A_3^1,A_3^2,B_3^1,B_3^2,C_3^1,C_3^2,D_3^1,D_3^2\}, \\ \nonumber
&
Z_4 = \{A_4^1,A_4^2,B_4^1,B_4^2,C_4^1,C_4^2,D_4^1,D_4^2\}
\end{align}
Subsequently, we focus on the users $\mathcal{S} = \{1,2,3\}$. Let us send them the following transmit vector
\begin{equation}
\mathbf{w}_{1,2}\tilde{X}_{1,2}+\mathbf{w}_{1,3}\tilde{X}_{1,3}+\mathbf{w}_{2,3}\tilde{X}_{2,3}
\end{equation}
where
$X_{1,2}=A_2^1\oplus B_1^1, %
X_{1,3}=A_3^1\oplus C_1^1, %
X_{2,3}=B_3^1\oplus C_2^1$.
This transmission should be such that $X_\mathcal{T}$ is received correctly at all users in $\mathcal{T}\subset \{1,2,3\}, |\mathcal{T}|=2$ . Let us call the corresponding common rate for coding each $X_\mathcal{T}$ as $R_{1,2,3}$. Then, since each minifile is of length $F/8$, the time needed for this transmission is
$T_{1,2,3}=\frac{F}{8}\frac{1}{R_{1,2,3}}$.
Now we consider the other $3$-subsets (subsets of size $3$) of users. For the subset $\mathcal{S} = \{1,2,4\}$ the transmitter sends
\begin{equation}
\mathbf{w}_{1,2}\tilde{X}_{1,2}+\mathbf{w}_{1,4}\tilde{X}_{1,4}+\mathbf{w}_{2,4}\tilde{X}_{2,4}
\end{equation}
where
$X_{1,2}=A_2^2\oplus B_1^2, %
X_{1,4}=A_4^1\oplus D_1^1, %
X_{2,4}=B_4^1\oplus D_2^1$
each coded with the rate $R_{1,2,4}$ and the corresponding transmission time is $T_{1,2,4}=\frac{F}{8}\frac{1}{R_{1,2,4}}$. Please note that the subset $\{1,2\}$ appears for the second time, and thus the second minifiles are used for the coding. The other subsets $\{1,4\}$, and $\{2,4\}$ have not yet appeared and the first minifiles are still not transmitted. %
For the subsets $\mathcal{S} =\{1,3,4\}$ and $\mathcal{S} =\{2,3,4\}$ the transmitter sends
\begin{align}
\mathbf{w}_{1,3}\tilde{X}_{1,3}+\mathbf{w}_{1,4}\tilde{X}_{1,4}+\mathbf{w}_{3,4}\tilde{X}_{3,4} \\
\mathbf{w}_{2,3}\tilde{X}_{2,3}+\mathbf{w}_{2,4}\tilde{X}_{2,4}+\mathbf{w}_{3,4}\tilde{X}_{3,4}
\end{align}
respectively, where
$X_{1,3}=A_3^2\oplus C_1^2, \
X_{1,4}=A_4^2\oplus D_1^2, \
X_{3,4}=C_4^1\oplus D_3^1$
are coded with the rate $R_{1,3,4}$ with the corresponding transmission time $T_{1,3,4}=\frac{F}{8}\frac{1}{R_{1,3,4}}$, while %
$X_{2,3}=B_3^2\oplus C_2^2, %
X_{2,4}=B_4^2\oplus D_2^2, %
X_{3,4}=C_4^2\oplus D_3^2$
are coded with the rate $R_{2,3,4}$ and
$T_{2,3,4}=\frac{F}{8}\frac{1}{R_{2,3,4}}$.
Since these transmissions are done in different time slots, the \emph{Symmetric Rate Per User} of this example is
\begin{align}
&\frac{F}{T_{1,2,3}+T_{1,2,4}+T_{1,3,4}+T_{2,3,4}} \\ \nonumber
&
=8\left(\frac{1}{R_{1,2,3}}+\frac{1}{R_{1,2,4}}+\frac{1}{R_{1,3,4}}+\frac{1}{R_{2,3,4}}\right)^{-1}.
\end{align}
The beamforming vectors are optimized separately to maximize the symmetric rate for each transmission interval. For each subset $\mathcal{S}$ the formulation is exactly the same as the one in \textit{Scenario 1}. The difference is that in this scenario we have potentially more antennas available ($L\geq 3$) allowing for further improved multicast beamforming performance. %

\subsection{Scenario 4: Simple Linear TX-RX strategy, \rev{$\alpha=3$ and $\beta=1$}}
In \textit{Scenarios 1--3}, each user is allocated with a number of parallel streams that need to be decoded using SIC receiver structure. %
In this example, in contrast, we consider the same setting as in \textit{Scenarios 2-3} with $L\geq3$, $K=4$, $N=4$, $M=1$ but no overlap is allowed among user groups served by multiple multicast messages transmitted in parallel\rev{, i.e., we use parameters $\beta=1$ and $\alpha=3$}. This leads to a simpler TX-RX strategy where all $6$ multicast streams introduced in \textit{Scenario 2} %
are delivered across three orthogonal time intervals/slots, instead of transmitting all in parallel as in~\eqref{eq:S2_signal_vector}. In time slots 1--3, the multicast beamforming vectors are generated as
$\mathbf{w}_{1,2}(A_2 \oplus B_1)+\mathbf{w}_{3,4}(C_4 \oplus D_3)$,
$\mathbf{w}_{1,3}(A_3 \oplus C_1)+\mathbf{w}_{2,4}(B_4 \oplus D_2)$
and
$\mathbf{w}_{1,4}(A_4 \oplus D_1)+\mathbf{w}_{2,3}(B_3 \oplus C_2)$,
respectively.%

In each time slot, all $4$ users are served with $2$ parallel multicast streams. Each stream causes inter-stream interference to $2$ other users not included in the given multicast group. Therefore, the BS, equipped at least with $3$ antennas, has enough spatial degrees of freedom to manage the inter-stream interference between multicast streams. 
The beamforming vectors are optimized separately to maximize the symmetric rate $R_C(i)$  for each transmission interval $i$.
Thus, the corresponding time to deliver the multicast messages containing $F/4$ fractions of the files in time slot $i$ is
$T(i)=\frac{F}{4}\frac{1}{R_C(i)}$.
Since these transmissions are done in $3$ different time slots, the overall \emph{Symmetric Rate Per User} of this scheme is
\begin{align}
\frac{F}{\sum_{i=1,2,3} T(i)} =4\left(\sum\nolimits_{i=1,2,3} R_C(i)^{-1}\right)^{-1}.
\end{align}
As will be shown in Section~\ref{sec:numexp}, the scheme provides the same overall DoF (slope) as the original scheme in \textit{Scenario 2}, but with a constant gap at high SNR due to simplified TX-RX processing.

As no overlap is allowed, each user decodes a single multicast message in a given time slot. Therefore, neither SIC receiver nor MAC rate region constraints are needed in the problem formulation unlike in \textit{Scenario 2}. As a result, the achievable rate is uniquely defined by the SINR of the received data stream.  Let us define $\gamma_C(i)$ to be the common symmetric SINR for all users served in time slot $i$ such that $R_C(i) = \log(1 + \gamma_C(i))$. 
The multigroup multicast beamformer optimization problem for $i$th timeslot can be then expressed as the following common SINR maximization problem:
\begin{equation}
\begin{array}{rl}
\label{prob:linear_problem}
\ds
\underset{\substack{\gamma_C(i), \M{w}_{\mathcal{T}}}}{\max} & 
\ds \ \gamma_C(i)
\\
\mathrm{s.\ t.} 
& \ds \gamma_C(i) \leq \frac{|\M{h}_k\herm\M{w}_{\mathcal{T}}(i)|^2}{|\M{h}_k\herm\M{w}_{\mathcal{\bar{T}}}(i)|^2 + N_0}, \\
&
\forall k \in \mathcal{T}, \mathcal{T} \in \mathsf{P}(i),  \mathcal{\bar{T}}  \in \mathsf{P}(i) \setminus \mathcal{T},
\\
& \sum_{\mathcal{T}\in \mathsf{P}(i)}
\|\M{w}_{\mathcal{T}}(i)\|^2 \leq \text{SNR}
\text{.}
\end{array}
\end{equation}
where $\mathsf{P}(1) = \{\{1,2\},\{3,4\}\}$, $\mathsf{P}(2) = \{\{1,3\},\{2,4\}\}$ and $\mathsf{P}(3) = \{\{1,4\},\{2,3\}\}$.
The resulting problem is a multi-group multicast beamforming for common SINR maximization and several solutions exist, for example via semidefinite relaxation (SDR) of beamformers and solving (iteratively via bisection) as a semidefinite program (SDP)~\cite{karipidis2008quality}. 
Here, instead, we adopt the SCA solution from ~\cite{Venkatraman-Tolli-Juntti-Tran-TSP17}, based on which~\eqref{prob:linear_problem} can be solved efficiently as a series of second order cone programs.  Unlike the SDP based designs, the SCA technique solves for beamformers directly, thereby avoiding the need for any randomization procedure if rank-1 beamformers are to be recovered from the SDR solutions~\cite{Venkatraman-Tolli-Juntti-Tran-TSP17}.

For example, by approximating the SINR constraints as in~\eqref{eq:sinr_constraint_before_sca}--\eqref{eq:sca}, the common SINR for time slot 1, $\gamma_C(1)$ can be solved (for a given approximation point $\bar{\M{w}}_{1,2}, \bar{\M{w}}_{3,4}, \bar{\gamma}_C(1)$ and by omitting the slot index $i$) as
\begin{equation}
\begin{array}{rl}
\label{prob:dlprob}
\ds
\underset{\substack{\gamma_C, \M{w}_{1,2},\M{w}_{3,4}}}{\max} & 
\ds
\gamma_C
\\
\mathrm{s.\ t.} 
&\ds \mathcal{L}({\M{w}_{1,2}, \M{w}_{3,4}, \M{h}_{1}, \gamma_C}) \geq 
|\M{h}_1\herm \M{w}_{3,4}|^2 + N_0,\\
&
\ds \mathcal{L}({\M{w}_{1,2}, \M{w}_{3,4}, \M{h}_{2}, \gamma_C}) \geq 
|\M{h}_2\herm \M{w}_{3,4}|^2 + N_0,
\\
&\ds \mathcal{L}({\M{w}_{3,4}, \M{w}_{1,2}, \M{h}_{3}, \gamma_C}) \geq 
|\M{h}_3\herm \M{w}_{1,2}|^2 + N_0,\\
&
\ds \mathcal{L}({\M{w}_{3,4}, \M{w}_{1,2}, \M{h}_{4}, \gamma_C}) \geq 
|\M{h}_4\herm \M{w}_{1,2}|^2 + N_0,
\\
& \ds \|\M{w}_{1,2}\|^2 + \|\M{w}_{3,4}\|^2 \leq \text{SNR}
\text{.}
\end{array}
\end{equation}
where $\mathcal{L}({\M{w}_{\mathcal{T}}, \M{w}_{\mathcal{\bar{T}}}, \M{h}_{k}, \gamma_C})$ is given in~\eqref{eq:sca}.

\section{General Case Formulation, Algorithm, and Rate Analysis}\label{sec:general}
\rev{As mentioned in Section~\ref{sec:multicast_bf_cc}, 
in the multiserver scheme~\cite{Shariatpanahi2017,Shariatpanahi-Caire-Khalaj-TIT18,Shariatpanahi2016}, each user $k$ will have to decode $m_k=\binom{t+L-1}{ t }$ different messages, which grows exponentially when $K$, $L$, $N$ are increased with the same ratio (linearly if $t=KM/N=1$).
Thus, the total number of rate constraints in the beamformer optimization problem is $(t+L) ( 2^{m_k} - 1 )$.} %
	For example, the case $L=4$, $K=5$, $N=5$ and $M=1$ would require altogether $ \binom52 =10$ multicast messages and each user should be able to decode 4 multicast messages. Thus, the total number of rate constraints would be $
	5\times 15$ while the number of SINR constraints to be approximated would be $5 \times 4$.
	As an efficient way to reduce the complexity of the problem both at the transmitter and the receivers (with a certain performance loss at high SNR), we may limit the size of user subsets benefiting from multicast messages transmitted in parallel as in \textit{Scenario 3} or limit the overlap among the multicast messages as in \textit{Scenario 4},  reflected in parameters $\alpha$ and $\beta$, respectively. %

In the following, the general algorithm for the delivery phase for any set of parameters
$K$, $L$, $N$ and $M$  in Algorithm~\ref{Alg_Main} is described. Let us first provide a light description of the algorithm. Algorithm's inputs contains library contents $W_1,\dots,W_N$, user requests indices $d_1,\dots,d_K$, and the channel gain matrix $\mathbf{H} = [\mathbf{h}_1, \ldots, \mathbf{h}_K]$. In addition, the algorithm takes two design parameters $\alpha$ and $\beta$ which should be tuned based on the working $SNR$, and the required complexity of the involved optimization problem. 

\rev{ The  cache content placement phase is similar to the scheme~\cite{MaddahAli-2014} introduced in Section~\ref{sec:multicast_bf_cc}, where each file is split into ${K \choose t}$ subfiles. The only difference is that here we further need to split each subfile in \cite{MaddahAli-2014} into smaller fragments }such that the total number of minifiles is 
\begin{equation}\label{eq:subpacketization}
{K \choose t} {K-t-1 \choose \alpha-1}\frac{(\alpha-1)!}{(\delta-1)!(\beta-1)!(t+\beta)!^{\delta-1}}
\end{equation}
where $\delta := \frac{t+\alpha}{t+\beta} \in \mathbb{N}$. This further splitting is needed in order to allow different content to be transmitted in each additional time interval introduced due to parameters $\alpha$ and $\beta$, similarly to \textit{Scenario 3}. \rev{More thorough justification for the second and third terms of~\eqref{eq:subpacketization} and their dependence of  $\alpha$ and $\beta$ is given in the latter part of this section. Note that~\eqref{eq:subpacketization} is reduced to~\cite{Shariatpanahi2016} if $\alpha=\beta=L$, and~\cite{MaddahAli-2014} if $\alpha=\beta=1$.}  

\rev{As shown in~\cite{Shariatpanahi2016,Shariatpanahi-Caire-Khalaj-TIT18}, coded caching/multicasting and spatial multiplexing gains are additive and the maximum DoF is upper bounded by $t+L$ (or by $\min(t+L,K)$ if $L>K-t$).}
In the generalized scheme, instead of fixing the size of the subsets $\{\mathcal{S} \subseteq [K]\}$ to be $\min(t+L,K)$ as in \textit{Scenario 2}, we introduce a new integer parameter $\alpha$ bounded by 
\begin{equation}\label{eq:alpha}
1 \leq \alpha \leq \min(L,K-t)
\end{equation} 
and define the size of subsets $\{\mathcal{S} \subseteq [K]\}$  to be \rev{$t+\alpha\leq t+\min(L,K-t)$.} The parameter $\alpha$ \rev{controls the available spatial multiplexing gain and} has two main roles. First, it manages the trade-off between the spatial multiplexing and multicast beamforming/diversity gains due to \rev{optimized use of} multiple transmit antennas, and thus should be designed carefully at each $SNR$ to result in the maximum throughput. Second, it enables us to control the size of the MAC channel elements with respect to each user, and in turn, to control the optimization problem complexity for determining the beamforming vectors, as will be explained later. This generalization reduces to the baseline max-min SNR beamforming scheme if $\alpha=1$ (see~\eqref{eq:CCmaxminSNR1} for $K=3$).%

Fig.~\ref{fig:cc-subset} illustrates a possible partitioning of users into sets $\{\mathcal{S} \subseteq [K]\}$ in a scenario with  $K=N=5$, $L=4$ and $M=1$, and with $|\mathcal{S}|=t+\alpha=\{2,3,4\}$ ($\alpha = \{1,2,3\}$). In total, there can be $T= {5 \choose 4} = 5$ , $T= {5 \choose 3} = 10$ and  $T= {5 \choose 2} = 10$ subsets of sizes $|\mathcal{S}|=4$, $|\mathcal{S}|=3$ and $|\mathcal{S}|=2$, respectively.
In this example, every subset in $\{\mathcal{S}, |\mathcal{S}|=3 \}$ or  $\{\mathcal{S}, |\mathcal{S}|=4 \}$ corresponds to \textit{Scenario 1} or \textit{Scenario 2}, respectively, and the optimal multicast beamformers can be found by solving~\eqref{prob:equal_rate_approx} or~\eqref{eq:symrate_4users} (for corresponding $k\in \mathcal{S}$).  
\begin{figure}
	\centering
	\includegraphics[width=\columnwidth]{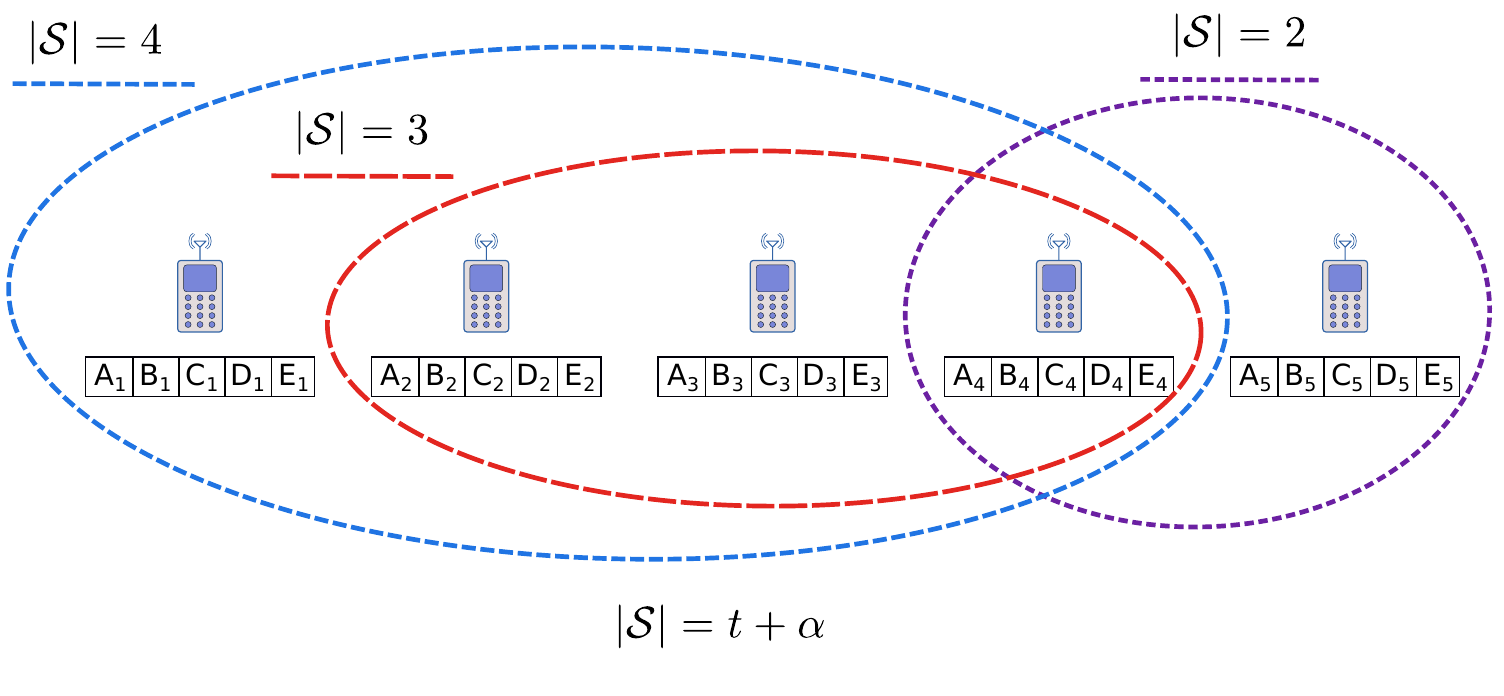}%
	\vspace{-0.3cm}	\caption{Example partitioning of users into sets $\{\mathcal{S} \subseteq [K]\}$ in a scenario with $K=N=5$,  $L=4$, $|\mathcal{S}| = t+\alpha=[2,3,4]$.}
	\label{fig:cc-subset}
\end{figure}

After the required initializations, the algorithm contains an outer loop which goes over all the $(t+\alpha)$-subsets of all the users $[K]$.
Let us now consider a scenario with $K=8$ users with $t=1$ and $\alpha=5$ depicted in Fig.~\ref{fig:cc-beta} and focus on one particular realization of these $(t+\alpha)$-subsets $\mathcal{S}=\{1,2,3,4,5,6\}$. For this specific set $\mathcal{S}$, the second loop goes over all possible partitionings of $\mathcal{S}$ into $(t+\beta)$-groups, which are collected in $\mathsf{P}$. Here, $\beta$, bounded by $1 \leq \beta \leq \alpha$, is another design parameter which controls the overlap among the multicast messages, i.e., the complexity of the beamformer design problem.
\begin{figure}
	\centering
	\includegraphics[width=\columnwidth]{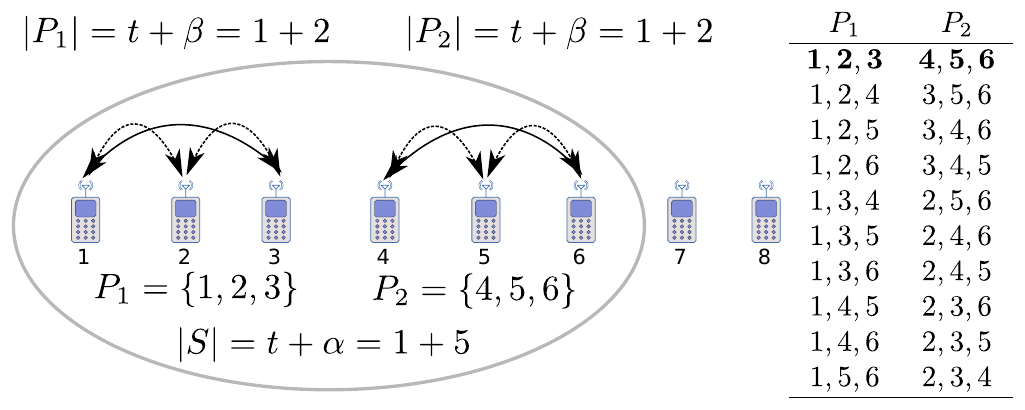}%
	\vspace{-0.3cm}	\caption{Example partitioning of users in $\mathcal{S}=\{1,2,3,4,5,6\}$  into $(t+\beta)$-groups in a scenario with $K=N=8$,  $L\geq 5$, $t=1$, $\alpha=5$ and $\beta=2$. }
	\label{fig:cc-beta}
\end{figure}

In the example of Fig.~\ref{fig:cc-beta}, $\beta=2$ is selected showing one possible partitioning of users in to $\delta=2$ groups $\mathsf{P}=\{\mathcal{P}_1, \mathcal{P}_2\}$ (of all $10$ possible partitionings shown also in the table), where $\mathcal{P}_1 = \{1,2,3\}$ and $\mathcal{P}_2 = \{4,5,6\}$.
Then, for a specific partitioning $\mathsf{P}$, we form the coded messages for all $(t+1)$-subsets of each group, for all the groups  $\mathcal{P}_i \in \mathsf{P}, i=1,...,\delta$. In this paper, the number of $t+\beta$ sets within each $t+\alpha$ set is restricted to integer values $\delta := \frac{t+\alpha}{t+\beta} \in \mathbb{N}$.\footnote{The generalization of $\delta$ is discussed in Remark 3.} In this example, there are $3$ parallel coded messages for every pair of users inside $\mathcal{P}_1$, and $3$ coded messages for every pair in $\mathcal{P}_2$, resulting in a total of $6$ coded messages. It should be noted that common messages for users in different groups are not allowed, which is the main ingredient behind controlling complexity of the beamformer design. In general, assuming $\delta \in \mathbb{N}$, there will be $\delta {t+\beta \choose t+1}$ coded messages involved for a fixed partitioning $\mathsf{P}$, while these $(t+1)$-subsets for multicast beamforming are collected in the collection of sets
 $\mathsf{\Omega}^{\mathcal{S},\mathsf{P}} := \bigcup_{i=1,\dots,\delta} \{\mathcal{T} \subseteq \mathcal{P}_i, |\mathcal{T}|=t+1\}$ for a specific $\mathcal{S}$ and $\mathsf{P}$.  The transmit vector $\underline{\mathbf{X}}(\mathcal{S},\mathsf{P})$ consists of all these coded messages multiplied by their corresponding beamformers.  
In the example shown in Fig.~\ref{fig:cc-beta}, the transmit vector $\underline{\mathbf{X}}(\mathcal{S},\mathsf{P})$ for $\mathcal{S}=\{1,2,3,4,5,6\}$,  $\mathcal{P}_1 = \{1,2,3\}$ and $\mathcal{P}_2 = \{4,5,6\}$ is generated as
 \begin{align}
\underline{\mathbf{X}}&(\mathcal{S},\mathsf{P}) = \mathbf{w}_{1,2}^{\mathcal{S},\mathcal{P}_1} \tilde{X}_{1,2}^{\mathcal{S},\mathcal{P}_1} + \mathbf{w}_{1,3}^{\mathcal{S},\mathcal{P}_1} \tilde{X}_{1,3}^{\mathcal{S},\mathcal{P}_1}+\mathbf{w}_{2,3}^{\mathcal{S},\mathcal{P}_1} \tilde{X}_{2,3}^{\mathcal{S},\mathcal{P}_1} \nonumber \\ 
&+ \mathbf{w}_{4,5}^{\mathcal{S},\mathcal{P}_2} \tilde{X}_{4,5}^{\mathcal{S},\mathcal{P}_2} + \mathbf{w}_{5,6}^{\mathcal{S},\mathcal{P}_2} \tilde{X}_{5,6}^{\mathcal{S},\mathcal{P}_2}+\mathbf{w}_{4,6}^{\mathcal{S},\mathcal{P}_2} \tilde{X}_{4,6}^{\mathcal{S},\mathcal{P}_2}.
 \end{align}
Note that here $\beta$ can control the number of coded messages aimed at each user. For example, if we allow $\beta=\alpha=5$, then there will be a total of ${6 \choose 2}=15$ coded messages transmitted in parallel, of which every user would need to decode $5$. By contrast, in the example scenario for $\beta=2$, there are in total $6$ parallel coded messages of which every user needs to decode~$2$. 
 
Finally, the beamformers are optimized to deliver each coded message to its intended users at the highest common rate, considering interference from other terms as well as noise. The optimum beamformers are denoted by $\{\mathbf{w}_\mathcal{T}^{\mathcal{S}, \mathsf{P}}, \mathcal{T} \in \mathsf{\Omega}^{\mathcal{S},\mathsf{P}}\}^*$ for a specific partitioning $\mathsf{P}$ of the set $\mathcal{S}$. The inner loop in the algorithm (line 8) ensures that the above procedure is repeated for all possible partitionings of a given $\mathcal{S}$ in a TDMA manner (for example in Fig.~\ref{fig:cc-beta}, all the $10$ possible partitionings in the table should be considered), and finally the outer loop repeats this process for all possible $(t+\alpha)$-subsets $\mathcal{S}$.

\begin{algorithm*}
	\caption{Interference aware Multi-Antenna Coded Caching\label{Alg_Main}}
	\begin{algorithmic}[1]
			\Procedure{DELIVERY}{$W_1,\dots,W_N$, $d_1,\dots,d_K$, $\mathbf{H} = [\mathbf{h}_1, \ldots, \mathbf{h}_K]$, $\alpha$, $\beta$}
			\State $t \gets MK/N$, \quad
			$\delta \gets \frac{t+\alpha}{t+\beta} \in \mathbb{N}$			%
			\ForAll{$\mathcal{S} \subseteq [K], |\mathcal{S}|=t+\alpha$}
			\ForAll{$\mathsf{P}=\{\mathcal{P}_i\}_{i=1,\dots,\delta}$: $\dot\bigcup_{i=1,\dots,\delta} \mathcal{P}_i=\mathcal{S}, |\mathcal{P}_i|=t+\beta$} \quad \quad %
			
			\State $\mathsf{\Omega}^{\mathcal{S},\mathsf{P}} \gets \bigcup_{i=1,\dots,\delta} \{\mathcal{T} \subseteq \mathcal{P}_i, |\mathcal{T}|=t+1\}$ 
			
			\ForAll{$\mathcal{T} \in \mathsf{\Omega}^{\mathcal{S},\mathsf{P}}$}
			\State $X_\mathcal{T} \gets \oplus_{k \in \mathcal{T}} 		
			\mathrm{NEW}(W_{{d_k},\mathcal{T}\backslash\{k\}}) $ \quad \quad \quad %
			\EndFor
			\State $\{\mathbf{w}_\mathcal{T}^{\mathcal{S}, \mathsf{P}}, \mathcal{T} \in \mathsf{\Omega}^{\mathcal{S},\mathsf{P}}\}^*=\arg  \max\limits_{\{\mathbf{w}_\mathcal{T}^\mathcal{S}, \mathcal{T} \in \mathsf{\Omega}^{\mathcal{S},\mathsf{P}}\}}   \min_{k \in \mathcal{S}} R^k_{MAC}\left(\mathcal{S},\mathsf{P},\{\mathbf{w}_\mathcal{T}^{\mathcal{S},\mathsf{P}}, \mathcal{T} \in \mathsf{\Omega}^{\mathcal{S},\mathsf{P}}\}\right)$  %
			\State $\underline{\mathbf{X}}(\mathcal{S},\mathsf{P}) \gets  \sum_{\mathcal{T} \in \mathsf{\Omega}^{\mathcal{S},\mathsf{P}}} \mathbf{w}_\mathcal{T}^{\mathcal{S},\mathsf{P}} \tilde{X}_\mathcal{T}$

			\State \textbf{transmit} $\underline{\mathbf{X}}(\mathcal{S},\mathsf{P})$ with the rate $\min_{k \in \mathcal{S}} R^k_{MAC}\left(\mathcal{S},\mathsf{P},\{\mathbf{w}_\mathcal{T}^{\mathcal{S},\mathsf{P}}, \mathcal{T} \in \mathsf{\Omega}^{\mathcal{S},\mathsf{P}}\}^*\right)$ 
			
			\EndFor
			\EndFor
			\EndProcedure
	\end{algorithmic}
\end{algorithm*}

The following theorem characterizes the achievable delivery rate of this algorithm. A detailed analysis of the algorithm elements, and the corresponding performance analysis is provided in the proof that follows. 
\begin{thm}\label{Th_sumrate_general}
	Algorithm \ref{Alg_Main} will result in the following symmetric rate

		\begin{equation}\label{eq:SymRateTh}
	R_{\mathrm{sym}}= \frac{F}{ 
		\sum_{\substack{\mathcal{S} \subseteq [K] \\ |\mathcal{S}|=t+\alpha}}  \sum_{\substack{\mathsf{P}=\{\mathcal{P}_i\} \\ \dot\bigcup \mathcal{P}_i=\mathcal{S} \\ |\mathcal{P}_i|=t+\beta}}  T_C^*(\mathcal{S},\mathsf{P})}
	\end{equation}
	where $T_C^*(\mathcal{S},\mathsf{P})$ is the optimized transmission time to the subset $\mathcal{S}$ for a specific partitioning $\mathsf{P}$. The outer sum is over all possible $t+\alpha$ sets $\mathcal{S}$ while the inner sum collects all disjoint unions $\dot\bigcup \mathcal{P}_i$ of $\mathcal{S}$ such that $|\mathcal{P}_i|=t+\beta$, and given $\delta := \frac{t+\alpha}{t+\beta} \in \mathbb{N}$. Each $T_C^*(\mathcal{S},\mathsf{P})$ is optimized over a set of multicast beamformers $\mathbf{w}_\mathcal{T}^{\mathcal{S},\mathsf{P}}, \mathcal{T} \in \mathsf{\Omega}^{\mathcal{S},\mathsf{P}}$
	\begin{align}\label{Eq:MulticastTime}&
T_C^*(\mathcal{S},\mathsf{P}) =\frac{F}{{K \choose t}{K-t-1 \choose \alpha-1} \frac{(\alpha-1)!}{(\delta-1)!(\beta-1)! (t+\beta)!^{\delta-1}}} \nonumber \\ 
&
\min_{\substack{\{\mathbf{w}_\mathcal{T}^{\mathcal{S},\mathsf{P}},  \mathcal{T} \in \mathsf{\Omega}^{\mathcal{S},\mathsf{P}},\\ \sum\limits_{ \mathcal{T} \in \mathsf{\Omega}^{\mathcal{S},\mathsf{P}}}\!\!\!\!{\|\mathbf{w}_\mathcal{T}^{\mathcal{S},\mathsf{P}}\|^2} \leq SNR\}}} \!\!\!\!\!%
\left[\min_{k \in \mathcal{S}} R^k_{MAC}\left(\mathcal{S},\mathsf{P},\{\mathbf{w}_\mathcal{T}^{\mathcal{S},\mathsf{P}} %
\}\right)\right]^{-1} %
\end{align}	
	where $R^k_{MAC}$ is the generalized stream specific rate expression for user $k$ and given as
	\begin{align}\label{Eq:MAC_general}
	&R^k_{MAC}\left(\mathcal{S},\mathsf{P},\{\mathbf{w}_\mathcal{T}^{\mathcal{S},\mathsf{P}}, \mathcal{T} \in \mathsf{\Omega}^{\mathcal{S},\mathsf{P}}\}\right) \nonumber \\ &
	= %
	\min_{\mathsf{B} \subseteq \mathsf{\Omega}_k^{\mathcal{S},\mathsf{P}}} \left[\frac{1}{|\mathsf{B}|}\log\left(1+\sum_{\mathcal{T} \in \mathsf{B}} \gamma_\mathcal{T}^k (\mathcal{S},\mathsf{P}) \right) \right]
	\end{align}
	and where
	\begin{equation}\label{eq:SINR_general}
	\gamma_\mathcal{T}^k(\mathcal{S},\mathsf{P})=\frac{|\mathbf{h}_k\herm \mathbf{w}_\mathcal{T}^{\mathcal{S},\mathsf{P}}|^2}{N_0+\sum_{\mathcal{\bar{T}} \in \bar{\mathsf{\Omega}}_k^{\mathcal{S},\mathsf{P}}} |\mathbf{h}_k\herm \mathbf{w}_\mathcal{\bar{T}}^{\mathcal{S},\mathsf{P}}|^2}
	\end{equation}
	and
	\begin{align}
	\mathsf{\Omega}^{\mathcal{S},\mathsf{P}} &
	:= \bigcup_{i=1,\dots,\delta} \{\mathcal{T} \subseteq \mathcal{P}_i, |\mathcal{T}|=t+1\}, \nonumber \\
	\mathsf{\Omega}_k^{\mathcal{S},\mathsf{P}} &
	:= \{\mathcal{T} \in \mathsf{\Omega}^{\mathcal{S},\mathsf{P}} | \ k \in \mathcal{T}\}, \ %
	\bar{\mathsf{\Omega}}_k^{\mathcal{S},\mathsf{P}} %
	:= \mathsf{\Omega}^{\mathcal{S},\mathsf{P}} \backslash \mathsf{\Omega}_k^{\mathcal{S},\mathsf{P}}.
	\end{align}
The SINR expressions in~\eqref{eq:SINR_general} are non-convex, and hence, they need to be relaxed and approximated in a successive manner, similarly to~\eqref{eq:sinr_constraint_before_sca}--\eqref{eq:sca}. First, \eqref{eq:SINR_general} is relaxed as 
\begin{align}\label{eq:sinr_general_before_sca}
 N_0+\sum\nolimits_{\mathcal{\bar{T}} \in \bar{\mathsf{\Omega}}_k^{\mathcal{S},\mathsf{P}}} |\mathbf{h}_k\herm \mathbf{w}_\mathcal{\bar{T}}^{\mathcal{S},\mathsf{P}}|^2 &\leq 
\frac{\sum_{\mathcal{\bar{T}} \in \mathcal{T} \bigcup \bar{\mathsf{\Omega}}_k^{\mathcal{S},\mathsf{P}}} %
|\mathbf{h}_k\herm \mathbf{w}_\mathcal{\bar{T}}^{\mathcal{S},\mathsf{P}}|^2 + N_0}
{1 +\gamma_\mathcal{T}^k(\mathcal{S},\mathsf{P})}
\end{align}
Now, the R.H.S of~\eqref{eq:sinr_general_before_sca} is a convex quadratic-over-linear function and it can be linearly approximated and lower bounded as (sets $\mathcal{S},\mathsf{P}$ omitted)
\begin{align}
	\label{eq:sca_general}
	\ds \mathcal{L}(\M{w}_{\mathcal{T}},& \M{w}_{\mathcal{\bar{T}}}, \M{h}_{k}, \gamma_\mathcal{T}^k) \triangleq \Bigg(
	\ds \sum_{\mathcal{\bar{T}} \in \mathcal{T} \bigcup \bar{\mathsf{\Omega}}_k^{\mathcal{S},\mathsf{P}}} \!\!\!\!\!\!\!\! |\M{h}_k\herm \bar{\M{w}}_{\mathcal{\bar{T}}}|^2 + N_0  \\ &
	\ds - 2 \!\!\!\!\!\!\!\!\sum_{\mathcal{\bar{T}} \in \mathcal{T} \bigcup \bar{\mathsf{\Omega}}_k^{\mathcal{S},\mathsf{P}}} \!\!\!\!\!\!\!\!  \mathbb{R}\left(
	\bar{\M{w}}_{\mathcal{\bar{T}}}\herm \M{h}_k \M{h}_k\herm 
	\left(\M{w}_{\mathcal{\bar{T}}} - \bar{\M{w}}_{\mathcal{\bar{T}}}\right)
	\right)  \nonumber \\
	&+ \frac{\!\sum_{\mathcal{\bar{T}} \in \mathcal{T} \bigcup \bar{\mathsf{\Omega}}_k^{\mathcal{S},\mathsf{P}}} %
	|\M{h}_k\herm \bar{\M{w}}_{\mathcal{\bar{T}}}|^2 + N_0}{1 + \bar{\gamma}_\mathcal{T}^k} 
	\left(\gamma_\mathcal{T}^k- \bar{\gamma}_\mathcal{T}^k\right) \rev{\Bigg) \frac{1}{1 + \bar{\gamma}_\mathcal{T}^k}} \nonumber
\end{align}
where $\bar{\M{w}}_{\mathcal{T}}$ and $\bar{\gamma}_\mathcal{T}^k$ denote the fixed values (points of approximation) for the corresponding variables from the previous iteration.

\end{thm}

Before going to the proof of Theorem~\ref{Th_sumrate_general}, let us revisit the simple scenarios introduced in Section~\ref{sec:multicast_bf_cc} and relate each of them to the generic algorithm above. By inserting the parameters listed below into~\eqref{eq:SymRateTh}--\eqref{Eq:MAC_general}, the corresponding scenario specific symmetric rate expressions given in Section~\ref{sec:multicast_bf_cc} can be recovered. 
\begin{itemize}
	\item \textit{Scenario 1}: $\alpha=2$, $\beta=2$, $\delta=1$, $\mathcal{S}=\mathsf{P}=\{1,2,3\}$
	\item \textit{Scenario 2}: $\alpha=3$, $\beta=3$, $\delta=1$,   $\mathcal{S}=\mathsf{P}=\{1,2,3,4\}$
	\item \textit{Scenario 3}: $\alpha=2$, $\beta=2$, $\delta=1$, $\mathcal{S}=\mathsf{P} \subset [4],|\mathcal{S}|=\alpha+1=3 $
	\item \textit{Scenario 4}: $\alpha=3$, $\beta=1$, $\delta=2$, $\mathcal{S}=\{1,2,3,4\}$, $\mathsf{P}(1)=\{\mathcal{P}_1(1), \mathcal{P}_2(1)\} = \{\{1,2\},\{3,4\}\}$, $\mathsf{P}(2) = \{\{1,3\},\{2,4\}\}$ and $\mathsf{P}(3) = \{\{1,4\},\{2,3\}\}$
\end{itemize}

The proof of Theorem \ref{Th_sumrate_general} is given in the following.

\begin{proof}

	In the cache content placement phase, each file is divided into ${K \choose t}$ subfiles as follows
	\begin{equation}
	W_n=\{W_{n,\tau}, \tau \subset [K], |\tau|=t\}, 
	\end{equation}
	and each subfile is further divided into mini-files
	\begin{equation}
	W_{n,\tau}=\{W_{n,\tau}^j, j=1,\dots, \Gamma \} 
	\end{equation}	
	where
	\begin{equation}\label{eq:Gamma}
	\Gamma={K-t-1 \choose \alpha-1}\frac{(\alpha-1)!}{(\delta-1)!(\beta-1)!(t+\beta)!^{\delta-1}}.
	\end{equation}
	
In the original coded caching scheme of \cite{MaddahAli-2014}, there are ${K \choose t+1}$ coded messages (called \emph{coded sub-files}, each of size equal to a sub-file) which should be delivered to all $(t+1)$-subsets of users $[K]$, i.e., $X_\mathcal{T} := \oplus_{k \in \mathcal{T}} 		W_{{d_k},\mathcal{T}\backslash\{k\}}$ should be delivered to all members of $\mathcal{T}$ for all $\mathcal{T} \subseteq [K], |\mathcal{T}|=t+1$. Since in our construction (inner and outer loops in Algorithm \ref{Alg_Main}, each $(t+1)$-subset appears multiple times, we need to transmit smaller coded messages (called \emph{coded mini-files}, each of size equal to a mini-file) in each appearance, which ensures that delivering each coded mini-file provides the targeted users with \emph{fresh} (not transmitted before) mini-files they require. This is the main reason behind  dividing each subfile into $\Gamma$ mini-files. In order to do this, we define the operator $\mathrm{NEW}(.)$ which when operated on each sub-file returns the next \emph{fresh} mini-file of that sub-file, which then will be used in forming coded mini-files. More specifically we have 
	\begin{equation}
	\mathrm{NEW}(W_{n,\tau})=W^{j+1}_{n,\tau}
	\end{equation}
	if the last application of $\mathrm{NEW}$ on the sub file $W_{n,\tau}$ had returned $W^{j}_{n,\tau}$. Next, we describe how these tasks are fulfilled with the help of multi-antenna interference management.
	
	Let us focus on a specific $(t+\alpha)$-subset of the users, namely $\mathcal{S}$, and a specific partitiong of this subset, namely $\mathsf{P}=\{\mathcal{P}_i\}_{i=1,\dots,\delta}$: $\dot\bigcup_{i=1,\dots,\delta} \mathcal{P}_i=\mathcal{S}, |\mathcal{P}_i|=t+\beta$. Then, $\mathsf{\Omega}^{\mathcal{S},\mathsf{P}}$ is the collection of all $(t+1)$-subsets of $\mathcal{S}$, such that each subset is contained inside a group $\mathcal{P}_i$ of the partition. Then, sum of coded mini-files of these $(t+1)$-subsets with the corresponding beamformers will be transmitted to users in $\mathcal{S}$ in the form of the transmit signal 
	\begin{equation}
	\underline{\mathbf{X}}(\mathcal{S},\mathsf{P}) \gets \ds \sum_{\mathcal{T} \in \mathsf{\Omega}^{\mathcal{S},\mathsf{P}}} \mathbf{w}_\mathcal{T}^{\mathcal{S},\mathsf{P}} \tilde{X}_\mathcal{T}^{\mathcal{S},\mathsf{P}}
	\end{equation}
	where $\tilde{X}_\mathcal{T}^{\mathcal{S},\mathsf{P}}$ is ensured to be a coded mini-file combined of fresh mini-files for each involved user. Assume that all the involved coded mini-files are successfully received at their intended users. %
	Then, all the subsets $\mathcal{T} \in \mathsf{\Omega}^{\mathcal{S},\mathsf{P}}$ will receive one coded mini-file, containing a fresh mini-file for each user in $\mathcal{T}$. It can be easily verified that, if we go over all the possible $(t+\alpha)$-subsets and their corresponding partitionings, each $(t+1)$-subset of $[K]$ will appear $\Gamma$
	times (given in \eqref{eq:Gamma}), and due to the appropriate mini-file indexing, each user will be able to decode a fresh mini-file in each transmission shot. Thus, these coded mini-files constitute the whole coded subfile. As this is true for all the $(t+1)$-subset of $[K]$, all the original tasks of \cite{MaddahAli-2014} are fulfilled.
	
	It just remains to be proven that by transmitting $\underline{\mathbf{X}}(\mathcal{S},\mathsf{P})$ with the rate stated in Theorem \ref{Th_sumrate_general}, all the users in $\mathcal{S}$ will be able to decode their desired coded mini-files. Consider a user $k \in \mathcal{S}$, which happens to be in the group $\mathcal{P}_i$ of the partitioning $\mathsf{P}$. Then, it is clear that this user will be interested in the coded mini-files $X_{\mathcal{T}}^{N_{\mathcal{T}}}$ such that $\mathcal{T} \in \mathsf{\Omega}_k^{\mathcal{S},\mathsf{P}}$, and all the remaining coded mini-files $X_{\mathcal{T}}^{N_{\mathcal{T}}}, \mathcal{T} \in \bar{\mathsf{\Omega}}_k^{\mathcal{S},\mathsf{P}}$ will appear as interference to this user. Thus, this user faces a Gaussian MAC with $|\mathsf{\Omega}_k^{\mathcal{S},\mathsf{P}}|$ desired terms, $|\bar{\mathsf{\Omega}}_k^{\mathcal{S},\mathsf{P}}|$ interference terms, and a noise term. Clearly, by restricting the transmission rate to the achievable Gaussian  MAC rate in \eqref{Eq:MAC_general}, this user can decode all the desired terms with an equal rate. Since we are transmitting the common message of size $\frac{F}{{K \choose t}\Gamma }$ to the users in $\mathcal{S}$ at the rate of the worst user, all of them will be able to decode the file within the minimum delivery time given in \eqref{Eq:MulticastTime}. %
	
	Finally, since each user decodes one requested file at the end, the symmetric (per-user) rate of the proposed scheme will be $R_{\mathrm{sym}} = F/T$ where the total time $T$ can be derived as
	\begin{equation}
	T=\sum_{\substack{\mathcal{S} \subseteq [K] \\ |\mathcal{S}|=t+\alpha}}  \sum_{\substack{\mathsf{P}=\{\mathcal{P}_i\} \\ \dot\bigcup \mathcal{P}_i=\mathcal{S} \\ |\mathcal{P}_i|=t+\beta}}  T_C^*(\mathcal{S},\mathsf{P}) 
	\end{equation}	
	where $T_C^*(\mathcal{S},\mathsf{P})$ is the transmission time for the given subset $\mathcal{S}$ and partitioning $\mathsf{P}$. %
\end{proof}

The following Degrees of Freedom (DoF) analysis of the proposed scheme shows that the DoF only depends on $\alpha$ and it is independent of $\beta$. \rev{By choosing $\alpha=L$, we achieve a DoF shown to be order-optimal among one-shot linear schemes in \cite{Naderalizadeh2017} (For an information theoretic optimality analysis based on interference alignment techniques see \cite{Hachem2016}).}
\begin{cor}
		The DoF of the rate derived in the above theorem is 
		\begin{equation}
		DoF=\frac{t+\alpha}{K-t}=\frac{KM/N+\alpha}{K(1-M/N)}.
		\end{equation}
	\end{cor}
	\begin{proof}
		DoF is defined as
		\begin{align}\label{eq:DoF}
		DoF&=\lim_{SNR \rightarrow \infty} \frac{R_{\mathrm{sym}}}{\log SNR} \nonumber \\ &
		=\frac{F}{ 
			\ds \sum_{\substack{\mathcal{S} \subseteq [K] \\ |\mathcal{S}|=t+\alpha}}  \sum_{\substack{\mathsf{P}=\{\mathcal{P}_i\} \\ \dot\bigcup \mathcal{P}_i=\mathcal{S} \\ |\mathcal{P}_i|=t+\beta}}  \lim_{SNR \rightarrow \infty}\left(\log SNR \times T_C^*(\mathcal{S},\mathsf{P})\right)}\nonumber \\ 
		&\stackrel{(a)}=\frac{F}{ {K \choose t+\alpha} \frac{(t+\alpha)!}{\delta! (t+\beta)!^{\delta}} \ds \lim_{SNR \rightarrow \infty}\left(\log SNR \times T_C^*(\mathcal{S},\mathsf{P})\right)}\nonumber \\ 	
		&\stackrel{(b)}=\frac{{K \choose t}{K-t-1 \choose \alpha-1} \frac{(\alpha-1)!}{(\delta-1)!(\beta-1)! (t+\beta)!^{\delta-1}}}{ {K \choose t+\alpha} \frac{(t+\alpha)!}{\delta! (t+\beta)!^{\delta}} } \nonumber \\ 	
		& \qquad
		\lim_{SNR \rightarrow \infty}\frac{R^k_{MAC}\left(\mathcal{S},\mathsf{P},\{\mathbf{w}_\mathcal{T}^{\mathcal{S},\mathsf{P}}, \mathcal{T} \in \mathsf{\Omega}^{\mathcal{S},\mathsf{P}}\}\right)}{\log SNR} \nonumber \\ 
		&\stackrel{(c)}=\frac{{K \choose t}{K-t-1 \choose \alpha-1} \frac{(\alpha-1)!}{(\delta-1)!(\beta-1)! (t+\beta)!^{\delta-1}}}{ {K \choose t+\alpha} \frac{(t+\alpha)!}{\delta! (t+\beta)!^{\delta}} {t+\beta-1 \choose t}} \nonumber \\ 	
		&
		=\frac{t+\alpha}{K-t}=\frac{KM/N+\alpha}{K(1-M/N)}. 
		\end{align}
		where $(a)$ is due to that fact that the number of terms in the inner and outer summations are $\frac{(t+\alpha)!}{\delta! (t+\beta)!^{\delta}}$ and ${K \choose t+\alpha}$ respectively, and since $ \lim_{SNR \rightarrow \infty}\left(\log SNR \times T_C^*(\mathcal{S},\mathsf{P})\right)$ does not depend on particular $\mathcal{S}$ and $\mathsf{P}$, $(a)$ is valid for any of $\mathcal{S}$ and $\mathsf{P}$ indexed in the summations. Also $(b)$ follows from \eqref{Eq:MulticastTime} and $(c)$ is due to the fact that
		\begin{align}
		& \lim_{SNR \rightarrow \infty}\frac{R^k_{MAC}\left(\mathcal{S},\mathsf{P},\{\mathbf{w}_\mathcal{T}^{\mathcal{S},\mathsf{P}}, \mathcal{T} \in \mathsf{\Omega}^{\mathcal{S},\mathsf{P}}\}\right)}{\log SNR}\nonumber \\ 
		&=\lim_{SNR \rightarrow \infty}\min_{\mathsf{B} \subseteq \mathsf{\Omega}_k^{\mathcal{S},\mathsf{P}}} \left[\frac{1}{|\mathsf{B}|}\frac{\log\left(1+\sum_{\mathcal{T} \in \mathsf{B}} \gamma_\mathcal{T}^k (\mathcal{S},\mathsf{P})\right)}{\log SNR} \right] \nonumber \\ 	
		&
		= \frac{1}{\ds \max_{\mathsf{B} \subseteq \mathsf{\Omega}_k^{\mathcal{S},\mathsf{P}}} |\mathsf{B}|} %
		= \frac{1}{{t+\beta-1 \choose t}} \nonumber
		\end{align}
		which concludes the proof.
\end{proof}

Also, we characterize the results in \cite{Shariatpanahi2017} and the max-min SNR beamforming baseline scheme as a special cases of Theorem \ref{Th_sumrate_general} in the following remark.

\begin{remark}
	In Theorem \ref{Th_sumrate_general}, if we set $\alpha=\beta=\min(L,K-t)$ and the beamforming vectors are chosen based on the zero forcing principle, the interference terms vanish, and it reduces to the results of \cite{Shariatpanahi2017}. Furthermore, if we set $\alpha=\beta=1$, the result reduces to the baseline maxmin SNR beamforming scheme.
\end{remark}
Moreover, the complexity of the optimization problem is characterized in the following remark.
\begin{remark}
    All of the constraints involved can be rewritten as second-order cones (SOCs). The SINR and transmit power constraints are readily in SOC form. However, the MAC sum rate constraints involving exponents (as seen, e.g., in~\eqref{prob:equal_rate_socp}) require some additional steps for the complete SOC formulation~\cite{Alizadeh01second-ordercone}. In the general case, the complexity of the beamformer design \eqref{eq:SymRateTh} is largely dominated by the number of simultaneously transmitted messages, that is, the partitioning size $|\mathsf{\Omega}_k^{\mathcal{S},\mathsf{P}}| = {t+\beta-1 \choose t}$. The number of MAC rate region constraints increases exponentially with $\beta + t$. However, the size of each SOC constraint involved with the MAC region is fairly small. On the other hand, the complexity of the SINR constraints scales quadratically with $\alpha + 1$ and $L$~\cite{Lobo-Vandenberghe-Boyd-Lebret-98}. It should be noted that, the beamformer design can be split into ${K \choose \alpha + t}$ parallel problems, which greatly improves the optimization latency and individual problem complexity as $\alpha$ is decreased.
    The receiver complexity is mostly affected by parameter $\beta$, i.e., whether or not SIC is needed. From the receiver perspective, $\beta>1$ indicates the number of desired multicast messages decoded at each user using the SIC receiver structure.
\end{remark}
\begin{remark}
	The above discussion is for parameter values such that $t+\beta$ divides $t+\alpha$. In general, one can vary $\alpha$ and $\beta$ such that this condition holds true, however if it is not possible to ensure, a readily available option is always to set $\beta=\alpha$, which by choosing $\alpha=L$ will achieve full DoF. For other cases where $t+\beta$ does not divide $t+\alpha$, extending the above techniques to arrive at satisfactory finite-SNR performance is challenging due to the asymmetries arising in the combinatorial nature of the problem.  Therefore, this problem can be posed as an interesting topic for further research.
\end{remark}

\section{Numerical Examples}
\label{sec:numexp}
The numerical examples are generated for various combinations of parameters $L,K,N,M$ and $|\mathcal{S}|$, including \emph{Scenarios 1 -- 4}. 
The channels are considered to be i.i.d. complex Gaussian. The average performance
is attained over $500$ independent channel realizations. The SNR is defined as
$\frac{P}{N_0}$, where $P$ is the power budget and $N_0 = 1$ is the fixed noise floor. All the Matlab codes are available online at \url{https://github.com/kalesan/sim-cc-miso-bc}.

\begin{figure}
	\centering
	\includegraphics[width=\columnwidth]{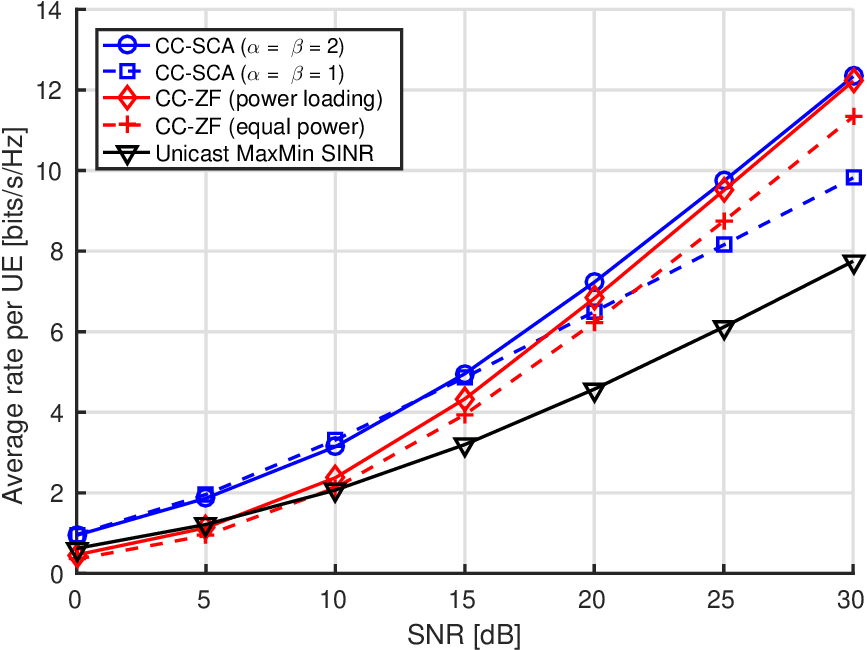}
	\vspace{-0.3cm}	\caption{\textit{Scenario 1}: $L=2$, $K=3$, $N=3$ and $M=1$.}
	\label{fig:K3}
\end{figure}
Fig.~\ref{fig:K3} shows the performance of the interference coordination with
CC in \emph{Scenario 1}, with $K=3$ users and $L=2$ antennas.  
It can be seen that the proposed CC multicast beamforming scheme via SCA, denoted as CC-SCA, achieves $3-5$ dB gain at low SNR as compared to the ZF with equal power loading~\cite{Shariatpanahi2017}. At high SNR, the ZF with optimal power loading in~\eqref{prob:equal_rate_ZF} achieves comparable performance while other schemes have significant performance gap. At low SNR regime, the simple MaxMin SNR multicasting with CC (labelled as 'CC-SCA ($\alpha=\beta=1$)') has similar performance as the CC-SCA scheme with full overlap between multicast streams ($\alpha=\beta=2$ ). This is due to the fact that, at low SNR, an efficient strategy for beamforming is to concentrate all available power to a single (multicast) stream at a time and to serve different users/streams in TDMA fashion. %
Due to simultaneous global CC gain and inter-stream interference handling, both CC-SCA and CC-ZF schemes achieve an additional DoF, which was already shown (for high SNR) in~\cite{Shariatpanahi2016, Shariatpanahi2017}.  The unicasting scheme does not perform well in this scenario as it does not utilize the global caching gain (only the local cache).%

In Fig.~\ref{fig:K3-L3}, the number of transmit antennas is increased to $L
= 3$.  This provides more than $3$dB additional gain for the \rev{CC-SCA} at
low SNR, when compared to the $L = 2$ antenna scenario, while the DoF is the same for
all the compared schemes. The optimal ZF multicast beamformer solution is no longer trivial, as the additional
antenna makes the interference free signal space two-dimensional for the ZF schemes.
A heuristic solution is
used where orthogonal projection is first employed
to get interference free signal space and then the strongest eigenvector of the stacked user channel matrix, projected to the null space, is used to get a sufficiently good direction within the interference free signal space. It can be seen that the ZF
scheme does achieve the same DoF as \rev{CC-SCA} method, but there is a constant
performance gap at high SNR. Interestingly, the \rev{CC-SCA} scheme with $L=2$
antennas has better performance than MaxMin SINR unicast with $L = 3$ antennas.
Both schemes have the same DoF, but the global caching gain is more beneficial than the additional spatial DoF of the unicast method.
\begin{figure}
	\centering
	\includegraphics[width=\columnwidth]{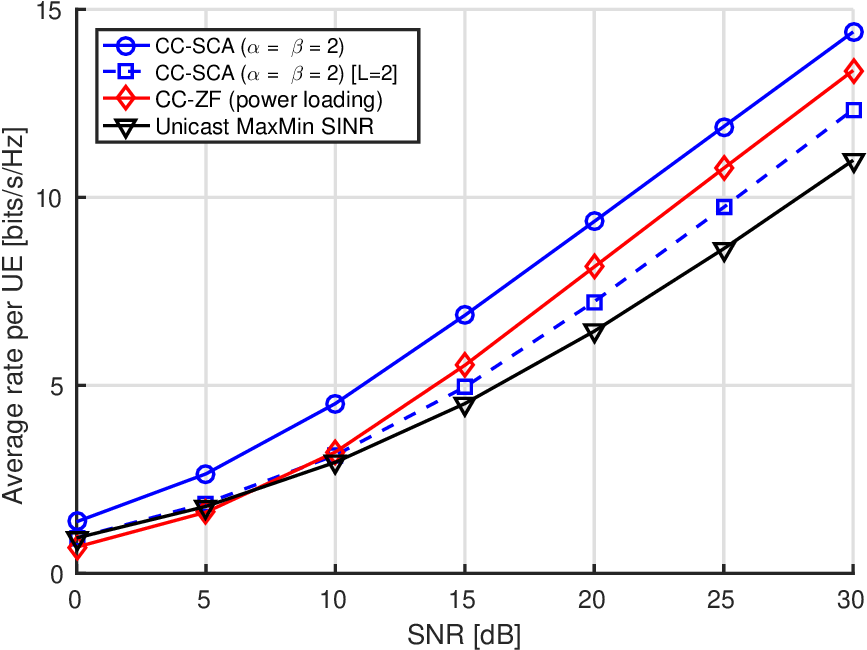}
	\vspace{-0.3cm}		\caption{\textit{Scenario 1}: $L=3$, $K=3$, $N=3$ and $M=1$.  }
	\label{fig:K3-L3}
\end{figure}

The performance of different schemes introduced in \emph{Scenarios 2--4} are illustrated in Fig.~\ref{fig:K4}. The CC-ZF scheme is not shown in Fig.~\ref{fig:K4}, but it is noted that the CC-SCA achieves $5-7$dB gain at low SNR, which is considerably more than in the less complex \emph{Scenario 1}. At high SNR, however, the CC-ZF with optimal power loading provides a comparable performance. %
\begin{figure}
	\centering
	\includegraphics[width=\columnwidth]{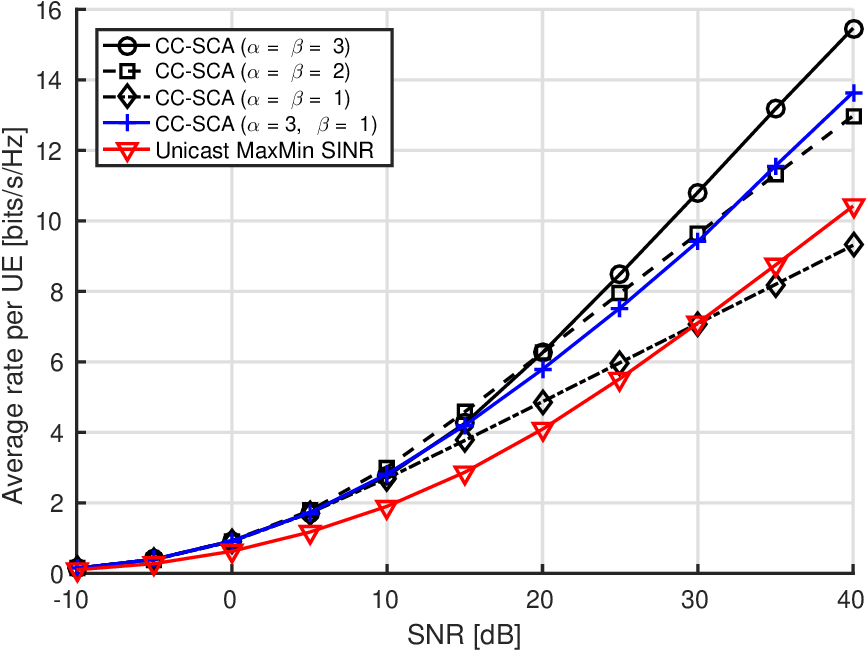}
	\vspace{-0.3cm}		\caption{Impact of parameters $\alpha$ and $\beta$ in \textit{Scenarios 2-4}: $L=3$, $K=N=4$, and $M=1$. }
	\label{fig:K4}
\end{figure}
The performance of the reduced subset method for $K=4$, $L = 3$ (\textit{Scenario 3}) are also shown in Fig.~\ref{fig:K4}. %
Similar to Fig.~\ref{fig:K3}, $\alpha=\beta=1 (|\mathcal{S}|=2)$ provides comparable performance at the low SNR regime. Interestingly, as the SNR is increased, $|\mathcal{S}|=3$ ($\alpha=\beta=2$) subset size outperforms the case with $|\mathcal{S}|=4$ ($\alpha=\beta=3$). %
Again, at lower SNR, it is better to focus the available power to fewer multicast streams transmitted in parallel. This will reduce the inter-stream interference and, at the same time, provide increased spatial degrees of freedom for multicast beamformer design.  All distinct user subsets $\mathcal{S} \subseteq [K]$ are served then in TDMA fashion. %

Fig.~\ref{fig:K4} illustrates also the performance of the proposed simple linear TX-RX multicasting scheme introduced in \textit{Scenario 4}.
The linear scheme labeled as 'CC-SCA ($\alpha=3$, $\beta=1$)' is able to serve 4 users simultaneously in each time slot with 3 antennas. Thus, it can provide the same degrees of freedom ($= 4$) at high SNR as the baseline CC-SCA scheme as well as its zero forcing (ZF) variant. However, there is about 3dB power penalty at high SNR due to less optimal TX-RX processing, but it still greatly outperforms the unicast reference case.

Fig.~\ref{fig:cc-subset} illustrates some subset selection possibilities for $K = 5$. For a six user ($K=6$) scenario shown in Fig.~\ref{fig:K6-subset}, there are four possible subset sizes $|\mathcal{S}| = [2, 3, 4,5]$ that can be used to reduce the serving set size for multicast transmission in $\mathcal{S}$. 
From Fig.~\ref{fig:K6-subset}, we can observe again that, by reducing the subset size to $|\mathcal{S}| = 5$ or $4$, the average symmetric rate per user can be even improved at medium SNR as compared to the case where all users are served simultaneously, i.e., $|\mathcal{S}|=6$. At high SNR region, however, the reduced subset cases become highly suboptimal as the spatial DoF for transmitting parallel streams is limited by $\alpha$. The high SNR slope for each curve in Fig.~\ref{fig:K6-subset} is equivalent to the user specific DoF given in~\eqref{eq:DoF}, ranging from $\frac{2}{5}$ ($\alpha=1$) to $\frac{6}{5}$ ($\alpha=5$).
From complexity reduction perspective, the multicast mode with the smallest subset size providing close to optimal performance should be selected. In Fig.~\ref{fig:K6-subset}, for example, subset sizes $|\mathcal{S}|=3$, $|\mathcal{S}|=4$, $|\mathcal{S}|=5$ could be used  up to 0 dB, 10 dB and 30 dB, respectively, for optimal performance-complexity trade-off.

%
%
%
%
%

	\begin{figure}
		\centering
		\includegraphics[width=\columnwidth]{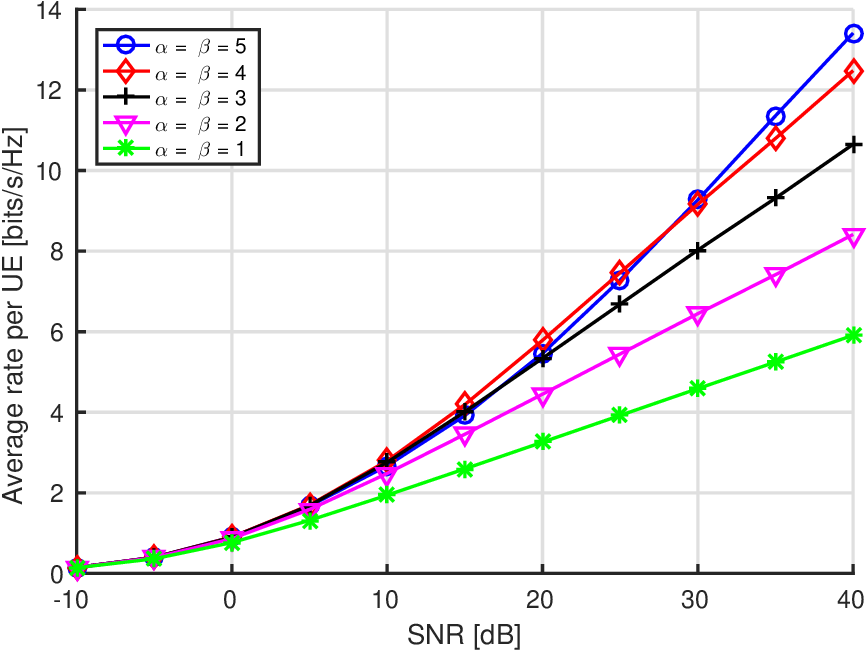}
			\caption{CC-SCA performance with $K=6$, $L=5$, $|S|=\alpha+1=[2,3,4,5,6]$.}\label{fig:K6-subset}
	\end{figure}

In Fig.~\ref{fig:K6-beta-subset}, the impact of parameter $\beta=[1,2,5]$ controlling the overlap among the parallel multicast messages is assessed with a fixed $\alpha=5$ for both SCA and ZF methods. The CC-ZF plots are generated by imposing zero-interference constraint similarly to \eqref{prob:equal_rate_ZF}.
The results with $\alpha=5$ and $\beta=1$ represent the case with no overlap, and hence, SIC is not required at the receivers. The multicast transmission is split into in total 15 time slots to cover all disjoint unions $\dot\bigcup \mathcal{T}$ of $\mathcal{S}=\{1,2,3,4,5,6\}$ such that $|\mathcal{T}|=t+1=2$. Furthermore, each subfile is split into $3$ mini-files in order to allow different contents to be transmitted in each subset $\mathsf{P}_i$. This is due to the fact that each user index pair $\mathcal{T}$ (e.g. $\mathcal{T}=\{1,2\}$) is repeated 3 times in distinct $\mathsf{P}_i$. In each time slot, all 6 users are served with 3 multicast streams transmitted in parallel. Thus, the BS, equipped at least with $5$ antennas, is able to manage the inter-stream interference between multicast streams. One spatial degree of freedom is used for delivering the multicast message to a user pair $\mathcal{T}$ while four degrees of freedom are needed to control the interference towards users $\mathcal{\bar{T}}  \in \mathsf{P}(i) \setminus \mathcal{T}$.
The case labeled as 'CC-SCA ($\alpha=5$, $\beta=2$) is an intermediate case between the linear scheme ($\beta=1$) and the fully overlapping case ($\beta=5$), allowing partial overlap among multicast messages transmitted in parallel. All possible 10 partitionings of size $t+\beta=3$ served in a TDMA fashion are shown in Fig.~\ref{fig:cc-beta}. In this case, each subfile must be further split into $4$ minifiles as each user index pair gets repeated in 4 different subsets $\mathcal{P}_i$

 \rev{The results in Fig.~\ref{fig:K6-beta-subset} verify the DoF analysis of Corollary 1 where the asymptotic DoF (slope) is shown to be independent of $\beta$ at high SNR region. However, the lower complexity  $\beta=1$ case in Fig.~\ref{fig:K6-beta-subset} suffers from $5$dB SNR penalty, which in turn can be alleviated by using a higher overlap ($\beta=2$) among parallel multicast streams. 
In general, a non-linear SIC structure has to be used at the receiver for $\beta>1$ in order to decode multiple parallel streams at each user.  In Fig. 8, the case with $\beta=L=5$ has the highest degree of flexibility for beamformer and power allocation since all $15$ multicast streams are transmitted in parallel. On the other hand, the optimization space for $\beta<5$ is much more constrained. For example, only $3$ multicast messages are sent in parallel in each transmit interval for $\beta=1$ and the SIC receiver is not needed at all. This 
translates to a constant penalty at high SNR when using $\beta<L$.
	
At low SNR, the performance loss from using highly suboptimal ZF criterion can be more than 10 dB at low SNR. Similarly, more than $150\%$ rate gain can be achieved by using the CC-SCA design at 10 dB SNR, for example. Furthermore, the performance impact of $\beta$ diminishes as the inter-stream interference is no longer dominant over the noise and all $\beta$ parametrizations provide almost identical performance. 
Asymptotically, the optimal transmit beamformers are reduced to SNR maximizing multicast beamformers, which can be simply obtained as weighted superposition of conjugate beamformers matching to the channels of each $t+1$-sized user subset $\mathcal{T}$. Consequently, the multicast beamformers specific to given $\mathcal{T}$ are essentially equivalent for any $\beta$ as they become independent of the inter-stream interference. Due to the maxmin rate objective, the beamformers are weighted such that the received signal level would be the same for all $k \in \mathcal{T}$. As a result, the effective channel gains $\mathbf{h}_k\herm\mathbf{w}_\mathcal{T}$ in the given symmetric scenario (equal path loss for all users) with $L=5$ antennas become (nearly) equal for all $k$ and $\mathcal{T}$.   
	 From the receiver perspective, this scenario resembles a classical symmetric downlink channel scenario where the superposition coding and SIC do not give any rate gain in comparison to a simple TDMA strategy~\cite[Section 6.2.1]{Tse-Viswanath-05}.
In general, $\beta$ should be selected as small as possible, since there is only a minor impact in the performance and, at the same time, there is a significant reduction in the computational complexity.}

	\begin{figure}
		\centering
		\includegraphics[width=\columnwidth]{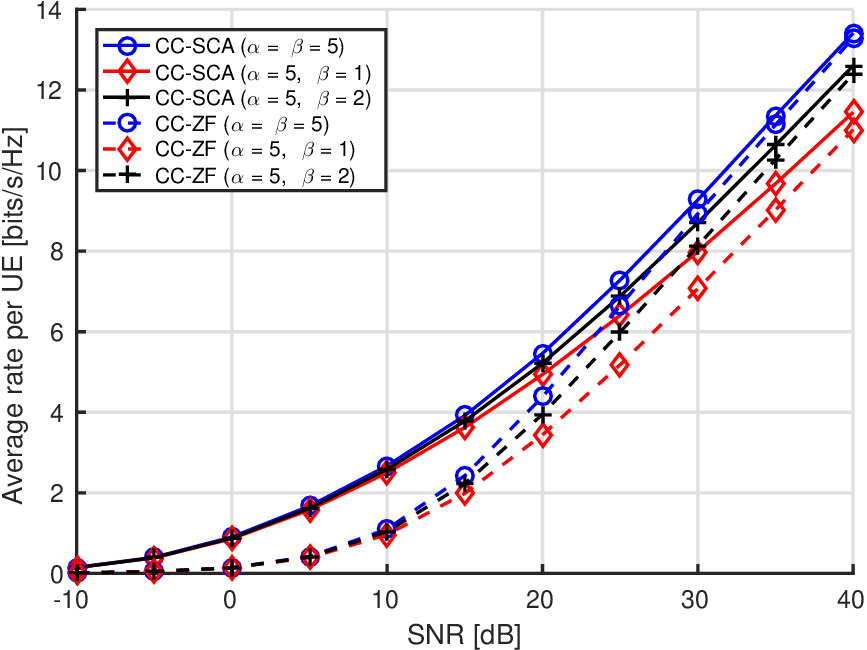}
			\caption{CC-SCA performance with $K=6$, $L=5$, $\alpha=5$ ($|S|=6$), $\beta = [1,2,5]$.} \label{fig:K6-beta-subset}
	\end{figure} 

\section{Conclusions}
Multicasting opportunities provided by caching at user terminal were utilized to devise an efficient multiantenna transmission with CC. General multicast beamforming strategies for content delivery with any values of the problem parameters, i.e., the number of users $K$, library size $N$, cache size $M$, and number of antennas $L$, size of the user subset $t+\alpha$, and the overlap among the multicast messages $\beta$ were employed, optimally balancing the detrimental impact of both noise and inter-stream interference from coded messages transmitted in parallel. %
Furthermore, the DoF was shown to only depend on $\alpha$ while being independent of $\beta$.
The schemes were shown to perform significantly better than several base-line schemes over the entire SNR region.

\bibliographystyle{IEEEtran}
\bibliography{IEEEabrv,conf_short,refs}

\begin{thebibliography}{10}
\providecommand{\url}[1]{#1}
\csname url@samestyle\endcsname
\providecommand{\newblock}{\relax}
\providecommand{\bibinfo}[2]{#2}
\providecommand{\BIBentrySTDinterwordspacing}{\spaceskip=0pt\relax}
\providecommand{\BIBentryALTinterwordstretchfactor}{4}
\providecommand{\BIBentryALTinterwordspacing}{\spaceskip=\fontdimen2\font plus
\BIBentryALTinterwordstretchfactor\fontdimen3\font minus
  \fontdimen4\font\relax}
\providecommand{\BIBforeignlanguage}[2]{{%
\expandafter\ifx\csname l@#1\endcsname\relax
\typeout{** WARNING: IEEEtran.bst: No hyphenation pattern has been}%
\typeout{** loaded for the language `#1'. Using the pattern for}%
\typeout{** the default language instead.}%
\else
\language=\csname l@#1\endcsname
\fi
#2}}
\providecommand{\BIBdecl}{\relax}
\BIBdecl

\bibitem{Cisco2019}
\BIBentryALTinterwordspacing
{Cisco Systems Inc.}, ``Cisco visual networking index: Global mobile data
  traffic forecast, update, 2017--2022 white paper,'' Cisco, Tech. Rep., Feb.
  2019. [Online]. Available:
  \url{https://www.cisco.com/c/en/us/solutions/collateral/service-provider/visual-networking-index-vni/white-paper-c11-738429.html}
\BIBentrySTDinterwordspacing

\bibitem{Shanmugam2013}
K.~Shanmugam, N.~Golrezaei, A.~G. Dimakis, A.~F. Molisch, and G.~Caire,
  ``Femtocaching: Wireless content delivery through distributed caching
  helpers,'' \emph{{IEEE} Trans. Inform. Theory}, vol.~59, no.~12, pp.
  8402--8413, Dec. 2013.

\bibitem{GolrezaeiD2D2012}
N.~Golrezaei, A.~F. Molisch, A.~G. Dimakis, and G.~Caire, ``Femtocaching and
  device-to-device collaboration: A new architecture for wireless video
  distribution,'' \emph{{IEEE} Commun. Mag.}, vol.~51, no.~4, pp. 142--149,
  Apr. 2013.

\bibitem{JiD2D2015}
M.~Ji, G.~Caire, and A.~F. Molisch, ``Wireless device-to-device caching
  networks: Basic principles and system performance,'' \emph{{IEEE} J. Select.
  Areas Commun.}, vol.~34, no.~1, pp. 176--189, Jan. 2016.

\bibitem{Bastug2014}
E.~Bastug, M.~Bennis, and M.~Debbah, ``Living on the edge: The role of
  proactive caching in {5G} wireless networks,'' \emph{{IEEE} Commun. Mag.},
  vol.~52, no.~8, pp. 82--89, Aug. 2014.

\bibitem{Gitzenis2013}
S.~Gitzenis, G.~S. Paschos, and L.~Tassiulas, ``Asymptotic laws for joint
  content replication and delivery in wireless networks,'' \emph{{IEEE} Trans.
  Inform. Theory}, vol.~59, no.~5, pp. 2760--2776, May 2013.

\bibitem{Liu2014}
A.~Liu and V.~K.~N. Lau, ``Cache-enabled opportunistic cooperative {MIMO} for
  video streaming in wireless systems,'' \emph{{IEEE} Trans. Signal
  Processing}, vol.~62, no.~2, pp. 390--402, Jan. 2014.

\bibitem{MaddahAli-2014}
M.~A. Maddah-Ali and U.~Niesen, ``Fundamental limits of caching,'' \emph{{IEEE}
  Trans. Inform. Theory}, vol.~60, no.~5, pp. 2856--2867, May 2014.

\bibitem{Wan-Tuninetti-Piantanida-ITW16}
{Kai Wan}, D.~{Tuninetti}, and P.~{Piantanida}, ``On the optimality of uncoded
  cache placement,'' in \emph{2016 IEEE Information Theory Workshop (ITW)},
  Sep. 2016, pp. 161--165.

\bibitem{Yu2018}
Q.~Yu, M.~A. Maddah-Ali, and A.~S. Avestimehr, ``The exact rate-memory tradeoff
  for caching with uncoded prefetching,'' \emph{{IEEE} Trans. Inform. Theory},
  vol.~64, no.~2, pp. 1281--1296, Feb. 2018.

\bibitem{Pedarsani2016}
R.~Pedarsani, M.~A. Maddah-Ali, and U.~Niesen, ``Online coded caching,''
  \emph{IEEE/ACM Transactions on Networking}, vol.~24, no.~2, pp. 836--845, Apr
  2016.

\bibitem{Karamchandi2016}
N.~Karamchandani, U.~Niesen, M.~A. Maddah-Ali, and S.~N. Diggavi,
  ``Hierarchical coded caching,'' \emph{{IEEE} Trans. Inform. Theory}, vol.~62,
  no.~6, pp. 3212--3229, Jun 2016.

\bibitem{Shariatpanahi2016}
S.~P. Shariatpanahi, S.~A. Motahari, and B.~H. Khalaj, ``Multi-server coded
  caching,'' \emph{{IEEE} Trans. Inform. Theory}, vol.~62, no.~12, pp.
  7253--7271, Dec 2016.

\bibitem{Zhang2017}
J.~Zhang and P.~Elia, ``Fundamental limits of cache-aided wireless {BC}:
  Interplay of coded-caching and {CSIT} feedback,'' \emph{{IEEE} Trans. Inform.
  Theory}, vol.~63, no.~5, pp. 3142--3160, May 2017.

\bibitem{Naderalizadeh2017}
N.~Naderializadeh, M.~A. Maddah-Ali, and A.~S. Avestimehr, ``Fundamental limits
  of cache-aided interference management,'' \emph{{IEEE} Trans. Inform.
  Theory}, vol.~63, no.~5, pp. 3092--3107, May 2017.

\bibitem{Naderalizadeh2017-2}
------, ``Cache-aided interference management in wireless cellular networks,''
  in \emph{Proc. IEEE Int. Conf. Commun.}, May 2017, pp. 1--7.

\bibitem{Tahmasbi2017}
M.~A.~T. Nejad, S.~P. Shariatpanahi, and B.~H. Khalaj, ``On storage allocation
  in cache-enabled interference channels with mixed {CSIT},'' in \emph{2017
  IEEE ICC Workshops}, May 2017, pp. 1177--1182.

\bibitem{MaddahIC2015}
M.~A. Maddah-Ali and U.~Niesen, ``Cache-aided interference channels,'' in
  \emph{Proc. IEEE Int. Symp. Inform. Theory}, Jun. 2015, pp. 809--813.

\bibitem{Cao2017}
Y.~Cao, M.~Tao, F.~Xu, and K.~Liu, ``Fundamental storage-latency tradeoff in
  cache-aided {MIMO} interference networks,'' \emph{{IEEE} Trans. Wireless
  Commun.}, vol.~16, no.~8, pp. 5061--5076, Aug. 2017.

\bibitem{Hachem2016}
J.~Hachem, U.~Niesen, and S.~N. Diggavi, ``Degrees of freedom of cache-aided
  wireless interference networks,'' \emph{IEEE Transactions on Information
  Theory}, vol.~64, no.~7, pp. 5359--5380, July 2018.

\bibitem{Roig2017}
J.~S.~P. Roig, S.~A. Motahari, F.~Tosato, and D.~G{\"{u}}nd{\"{u}}z,
  ``Fundamental limits of latency in a cache-aided 4x4 interference channel,''
  in \emph{2017 IEEE Information Theory Workshop (ITW)}, Nov. 2017, pp. 16--20.

\bibitem{Ji2016}
M.~Ji, G.~Caire, and A.~F. Molisch, ``Fundamental limits of caching in wireless
  {D2D} networks,'' \emph{{IEEE} Trans. Inform. Theory}, vol.~62, no.~2, pp.
  849--869, Feb 2016.

\bibitem{Shabani2016}
A.~Shabani, S.~P. Shariatpanahi, V.~Shah-Mansouri, and A.~Khonsari, ``Mobility
  increases throughput of wireless device-to-device networks with coded
  caching,'' in \emph{Proc. IEEE Int. Conf. Commun.}, May 2016, pp. 1--6.

\bibitem{Zhang2018}
J.~Zhang and O.~Simeone, ``Fundamental limits of cloud and cache-aided
  interference management with multi-antenna base stations,'' in \emph{Proc.
  IEEE Int. Symp. Inform. Theory}, June 2018, pp. 1425--1429.

\bibitem{Ngo2017}
K.~H. Ngo, S.~Yang, and M.~Kobayashi, ``Scalable content delivery with coded
  caching in multi-antenna fading channels,'' \emph{{IEEE} Trans. Wireless
  Commun.}, vol.~17, no.~1, pp. 548--562, Jan 2018.

\bibitem{Shariatpanahi2017}
S.~P. Shariatpanahi, G.~Caire, and B.~H. Khalaj, ``Multi-antenna coded
  caching,'' in \emph{Proc. IEEE Int. Symp. Inform. Theory}, Jun 2017, pp.
  2113--2117.

\bibitem{Piovano2017}
E.~Piovano, H.~Joudeh, and B.~Clerckx, ``On coded caching in the overloaded
  {MISO} broadcast channel,'' in \emph{Proc. IEEE Int. Symp. Inform. Theory},
  Jun 2017, pp. 2795--2799.

\bibitem{Shariatpanahi-Caire-Khalaj-TIT18}
S.~P. {Shariatpanahi}, G.~{Caire}, and B.~{Hossein Khalaj}, ``Physical-layer
  schemes for wireless coded caching,'' \emph{IEEE Transactions on Information
  Theory}, vol.~65, no.~5, pp. 2792--2807, May 2019.

\bibitem{Amiri2017}
M.~M. Amiri and D.~G{\"{u}}nd{\"{u}}z, ``Caching and coded delivery over
  gaussian broadcast channels for energy efficiency,'' \emph{IEEE Journal on
  Selected Areas in Communications}, vol.~36, no.~8, pp. 1706--1720, Aug 2018.

\bibitem{Bidokhti2017}
S.~S. Bidokhti, M.~Wigger, and A.~Yener, ``Benefits of cache assignment on
  degraded broadcast channels,'' in \emph{Proc. IEEE Int. Symp. Inform.
  Theory}, Jun. 2017, pp. 1222--1226.

\bibitem{Lampiris-Elia-2018}
E.~Lampiris and P.~Elia, ``Adding transmitters dramatically boosts
  coded-caching gains for finite file sizes,'' \emph{IEEE Journal on Selected
  Areas in Communications}, vol.~36, no.~6, pp. 1176--1188, June 2018.

\bibitem{Venkatraman-Tolli-Juntti-Tran-TSP17}
G.~Venkatraman, A.~T\"olli, M.~Juntti, and L.~N. Tran, ``Multigroup multicast
  beamformer design for {MISO-OFDM} with antenna selection,'' \emph{{IEEE}
  Trans. Signal Processing}, vol.~65, no.~22, pp. 5832--5847, Nov 2017.

\bibitem{Tolli-Shariatpanahi-Kaleva-Khalaj-ISIT18}
A.~T{\"{o}}lli, S.~P. Shariatpanahi, J.~Kaleva, and B.~H. Khalaj, ``Multicast
  beamformer design for coded caching,'' in \emph{Proc. IEEE Int. Symp. Inform.
  Theory}, Vail, CO, USA, Jun 2018.

\bibitem{Tolli-Shariatpanahi-Kaleva-Khalaj-WiOPT18}
------, ``Multicast mode selection for multi-antenna coded caching,'' in
  \emph{The 2018 International Workshop on Content Caching and Delivery in
  Wireless Networks (CCDWN)}, Shanghai, China, May 2018.

\bibitem{Tolli-Shariatpanahi-Kaleva-Khalaj-Asilomar18}
------, ``Linear multicast beamforming schemes for coded caching,'' in
  \emph{Proc. Annual Asilomar Conf. Signals, Syst., Comp.}, Pacific Grove, CA,
  USA, Oct 2018.

\bibitem{karipidis2008quality}
E.~Karipidis, N.~D. Sidiropoulos, and Z.-Q. Luo, ``{Quality of Service and
  Max-Min Fair Transmit Beamforming to Multiple Cochannel Multicast Groups},''
  \emph{{IEEE} Trans. Signal Processing}, vol.~56, no.~3, pp. 1268--1279, 2008.

\bibitem{Alizadeh01second-ordercone}
F.~Alizadeh and D.~Goldfarb, ``Second-order cone programming,''
  \emph{Mathematical Programming}, vol.~95, pp. 3--51, 2001.

\bibitem{Lobo-Vandenberghe-Boyd-Lebret-98}
M.~S. Lobo, L.~Vandenberghe, S.~Boyd, and H.~Lebret, ``Applications of
  second--order cone programming,'' \emph{Linear Algebra and Applications},
  vol. 284, pp. 193--228, Nov. 1998.

\bibitem{Tse-Viswanath-05}
D.~Tse and P.~Viswanath, \emph{Fundamentals of Wireless Communication}.\hskip
  1em plus 0.5em minus 0.4em\relax Cambridge, UK: Cambridge University Press,
  2005.

\end{thebibliography}

\begin{IEEEbiography}[{\includegraphics[width=1in,height=1.25in,keepaspectratio]{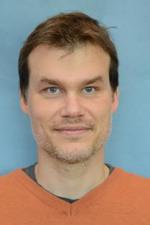}}]{Antti T\"olli}
    (M'08, SM'14) received the Dr.Sc. (Tech.) degree in electrical engineering from the University of Oulu, Oulu, Finland, in 2008. Before joining the Centre for Wireless Communications (CWC) at the University of Oulu, he worked for 5 years with Nokia Networks as a Research Engineer and Project Manager both in Finland and Spain. Currently, he holds an Associate Professor position with the University of Oulu. In May 2014, he was granted a five year (2014-2019) Academy Research Fellow post by the Academy of Finland. During the academic year 2015-2016, he visited at EURECOM, Sophia Antipolis, France, while from August 2018 to June 2019 he visited University of California - Santa Barbara, USA. He has authored numerous papers in peer-reviewed international journals and conferences and several patents all in the area of signal processing and wireless communications. His research interests include radio resource management and transceiver design for broadband wireless communications with a special emphasis on distributed interference management in heterogeneous wireless networks. He is currently serving as an Associate Editor for IEEE Transactions on Signal Processing.
\end{IEEEbiography}

\begin{IEEEbiography}[{\includegraphics[width=1in,height=1.25in,keepaspectratio]{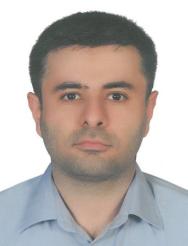}}]{Seyed Pooya Shariatpanahi} received the B.Sc., M.Sc., and Ph.D. degrees from the Department of Electrical Engineering, Sharif University of Technology, Tehran, Iran, in 2006, 2008, and 2013, respectively. Currently, he is an assistant professor at the School of Electrical and Computer Engineering at University of Tehran. Before joining University of Tehran, he was a researcher with the Institute for Research in Fundamental Sciences (IPM), Tehran, Iran. His research interests include information theory, network science, wireless communications, and complex systems. He was a recipient of the Gold Medal at the National Physics Olympiad in 2001.
\end{IEEEbiography}

\begin{IEEEbiography}[{\includegraphics[width=1in,height=1.25in,clip,keepaspectratio]{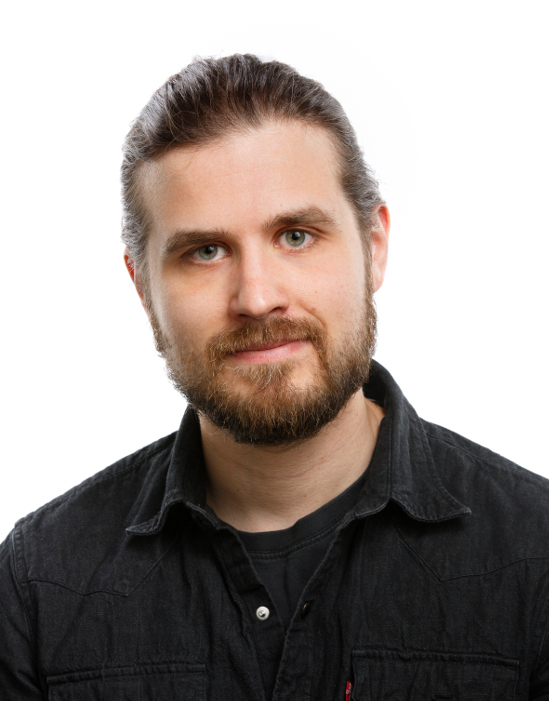}}]{Jarkko Kaleva}
    (S'11-M'18) received his Dr.Sc. (Tech.) degree in communications engineering from University of Oulu, Oulu, Finland in 2018 with distinction. In 2010, he joined Centre for Wireless Communications (CWC) at University of Oulu, Finland. He is co-founder of Solmu Technologies, where he is working as the chief software architect. His main research interests are in deep learning, structural analysis and nonlinear programming.
\end{IEEEbiography}

\begin{IEEEbiography}[{\includegraphics[width=1in,height=1.25in,clip,keepaspectratio]{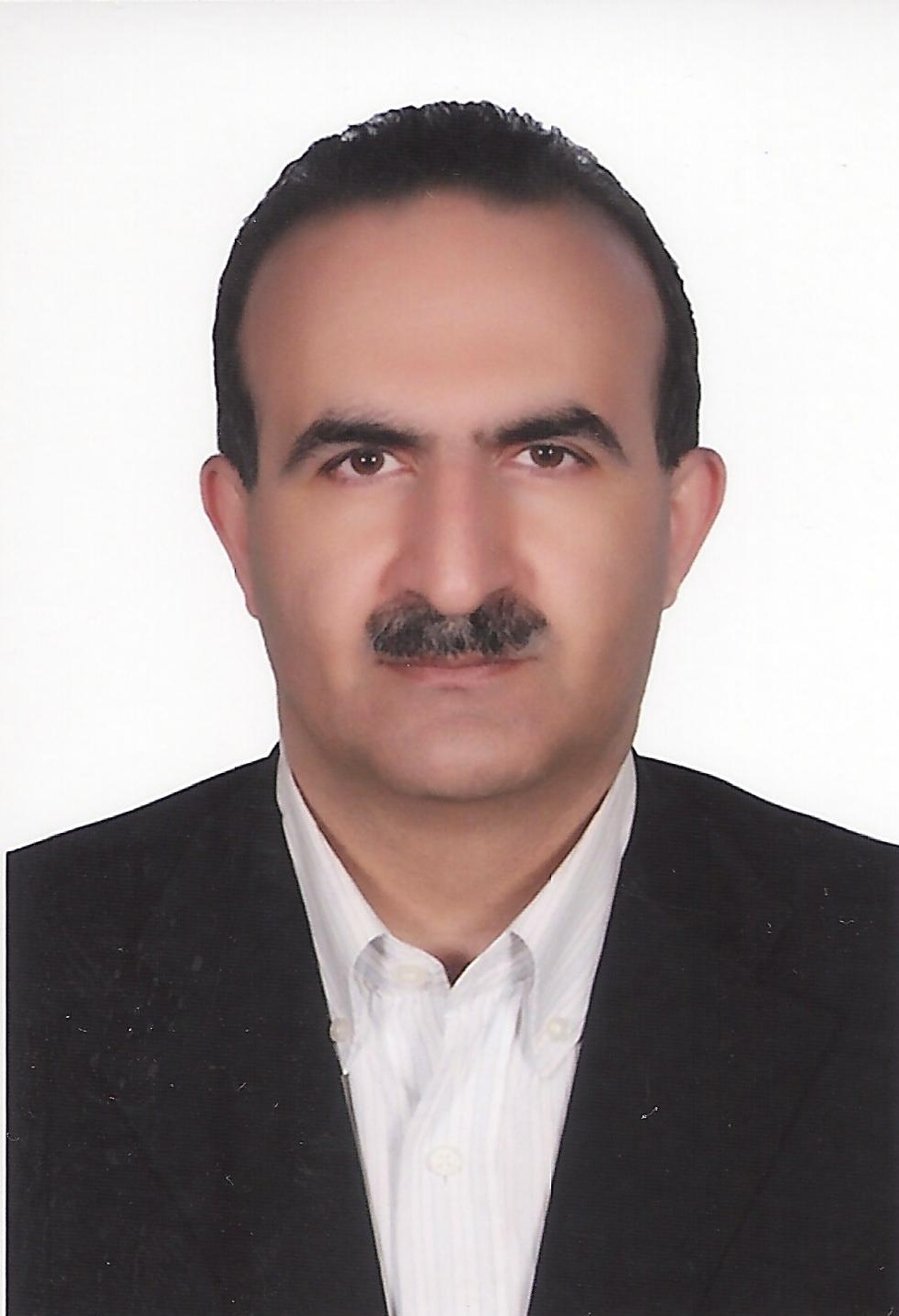}}] {Babak Hossein Khalaj} received the B.Sc. degree in Electrical Engineering from Sharif University of Technology, Tehran, Iran, in 1989, and M.Sc. and Ph.D. degrees in Electrical Engineering from Stanford University, Stanford, CA, USA, in 1993 and 1996, respectively. He is currently a Professor in the Department of Electrical Engineering at
Sharif University of Technology. He has been with the pioneering team at Stanford University where he was involved in adoption of multi-antenna arrays in mobile networks. He joined KLA-Tencor in 1995, as a Senior Algorithm Designer, involved on advanced processing techniques for signal estimation. From 1996 to 1999, he was with Advanced Fiber Communications and Ikanos Communications. Since then, he has been a Senior Consultant
in the area of data communications, and a Visiting Professor with CEIT, San Sebastian, Spain, from 2006 to 2007. He has co-authored many papers in signal processing and digital communications. He holds three U.S. patents
and was the recipient of the Alexander von Humboldt Fellowship from 2007 to 2008 and Nokia Visiting Professor Fellowship in 2018. 
\end{IEEEbiography}

\end{document}